\documentclass[runningheads,12pt]{llncs}
\usepackage{graphicx}
\usepackage{booktabs}
\usepackage{wrapfig,subcaption}
\usepackage{enumerate,setspace}
\usepackage{amsmath,amssymb,latexsym} 
\usepackage[margin=1.2in]{geometry}
\usepackage{etoolbox}
\usepackage{enumitem}
\usepackage{color}
\usepackage[colorlinks,linkcolor=blue,citecolor=magenta]{hyperref}
\usepackage[colorinlistoftodos,bordercolor=orange,backgroundcolor=orange!20,linecolor=orange,textsize=scriptsize]{todonotes}
\usepackage{soul}
\usepackage{microtype}
\usepackage{yhmath}
\usepackage{tikz}
\usetikzlibrary{arrows,shapes, calc, fit, positioning}
\usepackage{caption}
\usepackage{mathtools}
\usepackage{verbatim}
\usepackage{times}
\usepackage[boxed]{algorithm}
\usepackage[noend]{algorithmic}
\usepackage{mathtools}
\usepackage{cite} 
\usepackage[square,sort]{natbib}
\usetikzlibrary{decorations.pathreplacing,calligraphy}

\usepackage{subcaption}
\usetikzlibrary{positioning,chains,fit,shapes,calc}
\newcommand{\bfx}{\boldsymbol{x}}

\newcommand{\bfe}{\boldsymbol{e}} 

\newcommand{\bbR}{\mathbb{R}}
\newcommand{\bbN}{\mathbb{N}}

\newcommand{\argmax}{\mbox{argmax}}

\newcommand{\half}{\mathrm{half}}

\newcommand{\calI}{\mathcal{I}}

\newcommand{\ttT}{\mathsf{T}}
\newcommand{\ceil}[1]{\lceil #1 \rceil }
\newcommand{\floor}[1]{\lfloor #1 \rfloor }

\def\le{\leqslant}
\def\ge{\geqslant}

\begin{document}
\title{How to cut a discrete cake fairly}
%
%
\author{Ayumi Igarashi}
\authorrunning{Igarashi}
%
\institute{National Institute of Informatics, Japan,
\email{ayumi\_igarashi@nii.ac.jp}}
\maketitle              
\begin{abstract}
Cake-cutting is a fundamental model of dividing a heterogeneous resource, such as land, broadcast time, and advertisement space. 
In this study, we consider the problem of dividing a discrete cake fairly in which the indivisible goods are aligned on a path and agents are interested in receiving a connected subset of items. We prove that a connected division of indivisible items satisfying a discrete counterpart of envy-freeness, called {\em envy-freeness up to one good} (EF1), always exists for any number of agents $n$ with monotone valuations. Our result settles an open question raised by Bil\`o et al. (2019), who proved that an EF1 connected division always exists for the number of agents $n \le 4$. Moreover, the proof can be extended to show the following (1) ``secretive" and  (2) ``extra" versions: (1) for $n$ agents with monotone valuations, the path can be divided into $n$ connected bundles such that an EF1 assignment of the remaining bundles can be made to the other agents for any selection made by the “secretive agent”; (2) for $n+1$ agents with monotone valuations, the path can be divided into $n$ connected bundles such that when any ``extra agent” leaves, an EF1 assignment of the bundles can be made to the remaining agents.  
\end{abstract}

\section{Introduction}
Imagine a group of researchers scheduling time slots for meetings. Their preferences may be heterogeneous: for example, one researcher may prefer a morning meeting, whereas another may prefer an afternoon meeting. This situation raises the question: how can we allocate time slots fairly? This problem falls within the field of the well-known cake-cutting problem, in which the {\em cake}, often represented by the unit interval $[0,1]$, has to be divided between $n$ agents with different preferences. Here, the term ``cake” is a metaphor for a heterogeneous divisible resources, such as land or time. 

A central notion of fairness in the literature is {\em envy-freeness} \citep{Foley67}, which requires that each agent receives their personal best piece out of the allocated pieces. In this scenario, no agent wishes to replace their allocated portion with that of any other agent. 
The classical result shows that under
mild assumptions on the agents' preferences, there is an envy-free division offering each agent a connected piece \citep{Stromquist80,Su1999,Woodall1980}. Note that the connectivity constraint is crucial in various contexts, particularly when the resource has temporal or spacial structure. 

In many application domains, the resource may be indivisible. For example, time is often divided into discrete time units, such as scheduled shifts and research seminars. As another example, land may be divided into discrete land plots, due to geographical or historical constraints. A discrete version of the cake-cutting problem has been considered in several papers recently~\citep{BiloCaFl19,BouveretCeEl17,Marenco2014,Suksompong2019}. In this framework, the indivisible items are aligned on a path and each agent is allocated to a connected bundle of items. 

In allocation of indivisible resources, envy-freeness is not guaranteed. Indeed, in the case of one item and two agents, one agent necessarily receives nothing and therefore envies the other. Nevertheless, the objective can be naturally relaxed via approximations. An approximate notion of envy-freeness, called \emph{envy-freeness up to one good} (EF1), has been intensively investigated in recent years \citep{Budish11}. EF1 allows agents to envy other agents but the envy can be eliminated after removing one item from others' bundles. Several algorithms achieve EF1 within the standard setting of fair division of indivisible items~\citep{LiptonMaMo04,CaragiannisKuMo16}. For agents with monotone valuations, the envy-cycle algorithm in \citet{LiptonMaMo04} returns an EF1 division. For agents with monotone additive valuations, an allocation maximizing the Nash product of agents' valuations satisfies EF1 \citep{CaragiannisKuMo16}. Can we achieve EF1 under the connectivity constraint of a path? 

This question was partially answered by \citet{BiloCaFl19}. They showed that an EF1 connected division exists for any monotone valuations when there are at most four agents. They followed Su's approach \citep{Su1999} using Sperner's lemma: The possible divisions of the path can be encoded by the vertices of a triangulated $(n-1)$-dimensional simplex; see Figure~\ref{fig:Sperner} for an illustration with $n=3$. Each agent colors each vertex with the index of the most preferred bundle of the partition represented by each vertex. Sperner's lemma then implies the existence of a small simplex labeled with distinct agents and colored with different indices of bundles. Loosely speaking, this simplex corresponds to a sequence of ``similar” divisions, each of which satisfies different agents with different bundles. 
\citet{BiloCaFl19} developed a rounding technique of such a simplex that works for four or fewer agents. 
However, the existence of EF1 is unclarified when the number of agents $n$ is more than four, while a connected division satisfying the weaker fairness notion of EF2, requiring the envy to be bounded up to two items, exists for any number of agents \citep{BiloCaFl19}. 
A key issue in this proof lies in the left-right {\em symmetry} in the agents' evaluation. Since agents treat left-most and right-most items of a bundle symmetrically, they may face a conflict with another agent who want neighboring bundles; see Section~\ref{sec:potential} for the detailed discussion. 

To settle this open question, we show that an EF1 connected division exists for any number of agents with monotone valuations.   
Our proof adopts the rounding technique of a simplex, similar to that developed in \citet{BiloCaFl19}, but our final path division has left-right {\em asymmetry}. Namely, our rounding algorithm prioritizes an agent who wants a left bundle over an agent who wants a right bundle on the simplex. In our proof, each agent is pessimistic about obtaining the left-most item and optimistic about obtaining the right-most item, which ensures that the estimate of each agent $i$ on the $j$-th bundle of the final output is neither overly optimistic nor overly pessimistic.
In this way, we can successfully circumvent the difficulty arising when $n \geq 5$. In fact, we show the existence of a connected division satisfying a slightly stronger notion of EF1$_{\emph outer}$, which additionally requires that the envied bundles remain connected after removing an item.\footnote{The result of \citet{BiloCaFl19} also holds with EF2$_{\emph outer}$ for any number of agents with monotone valuations.} 

Exploiting the proof of this theorem, we further obtain the discrete analogs of the existential result concerning a secretive envy-free divisions of the cake \citep{Woodall1980,Asada2018,MeunierSu}: for any number $n$ of agents with monotone valuations over a path, one can divide the path into $n$ parts such that whichever part a secretive agent chooses, an EF1$_{\emph outer}$ assignment of the remaining bundles can be made to the other agents.

For two agents, such existence directly follows from a discrete version of the cut-and-choose protocol over the path: the first agent computes an EF1$_{\emph outer}$ connected division among two agents as if the other agent has the same valuation, and then the second agent (called the {\em secretive agent}) selects a preferred bundle, leaving the remainder for the first agent. Here, the valuation of one agent is sufficient to find a partition such that whichever part another agent chooses, the resulting assignment is EF1$_{\emph outer}$. 
More generally, our result shows that the valuations of $n-1$ agents suffice to find an EF1$_{\emph outer}$ connected division among $n$ agents. Note that without connectivity constraints, a secretive EF1 division is known to exist and can be computed in polynomial time when the agents have monotone valuations \citep{Arunachaleswaran2019}. However, our result is the first to show that the existential result holds in conjunction with connectivity requirements. 

Finally, we show the dual statement that for $n+1$ agents with monotone valuations, the path can be partitioned into $n$ connected subsets such that when any extra agent leaves, an EF1$_{\emph outer}$ assignment of the bundles can be made to the remaining agents, thereby establishing the discrete counterpart of the existence of an extra envy-free cake division, recently shown by \citet{MeunierSu}.  

An important application of our results is that for graph fair division proposed by \citet{BouveretCeEl17}. This setting captures, e.g., the division of road networks, where the items can be aligned on a graph and the agents value connected bundles of the items. Our results on a path apply to a wider class of \emph{traceable graphs} that admit a Hamiltonian path. Indeed, for such graphs, one can take a Hamiltonian path of the original graph and apply our result to obtain an EF1$_{\emph outer}$ connected division of the Hamiltonian path; the resulting division is trivially both EF1$_{\emph outer}$ and connected in the original graph.\footnote{This observation has already been made in \citep{BiloCaFl19}.} Thus, for any number of agents with monotone valuations, an EF1$_{\emph outer}$ connected division of a graph exists whenever the graph is traceable.

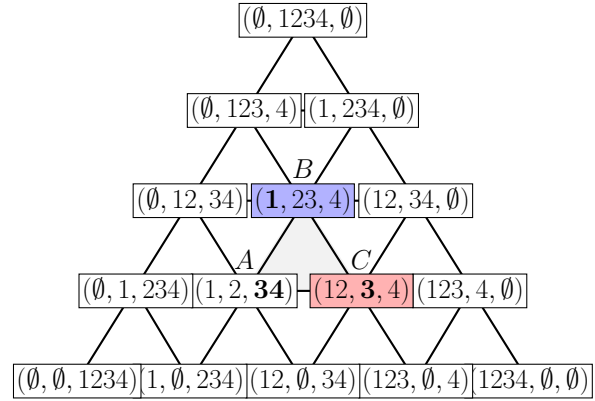
\begin{figure}[tbh]
\centering
    \begin{tikzpicture}[scale=0.6, transform shape]


    \filldraw[gray!10] (5,4) -- (3.7,2) -- (6.3,2) -- (5,4);

    \draw[thick] (0,0) -- (10,0) -- (5,8) -- (0,0);
    \draw[thick] (1.3,2) -- (8.7,2);
    \draw[thick] (2.5,4) -- (7.5,4);
    \draw[thick] (3.7,6) -- (6.3,6);

    \draw[thick] (1.3,2) -- (2.5,0) -- (6.3,6);
    \draw[thick] (2.5,4) -- (5,0) -- (7.5,4);
    \draw[thick] (3.7,6) -- (7.5,0) -- (8.7,2);

    \node[draw,fill=white, inner sep=2pt,minimum size=2pt] at (5,8) {\Large $(\emptyset,1234,\emptyset)$};
    \node[draw,fill=white,inner sep=2pt,minimum size=2pt] at (6.3,6) {\Large $(1,234,\emptyset)$};
    \node[draw,fill=white,inner sep=2pt,minimum size=2pt] at (7.5,4) {\Large $(12,34,\emptyset)$};
    \node[draw,fill=white,inner sep=2pt,minimum size=2pt] at (8.7,2) {\Large $(123,4,\emptyset)$};
    \node[draw,fill=white,inner sep=2pt,minimum size=2pt] at (10,0) {\Large $(1234,\emptyset,\emptyset)$};

    \node[draw,fill=white,inner sep=2pt,minimum size=2pt] at (3.7,6) {\Large $(\emptyset,123,4)$};
    \node at (5,4.7) {\Large $B$}; 
    \node[draw,fill=blue!30,inner sep=2pt,minimum size=2pt] at (5,4) {\Large $({\bf 1},23,4)$};
    \node at (6.3,2.7) {\Large $C$};
    
    \node[draw,fill=red!30,inner sep=2pt,minimum size=2pt] at (6.3,2) {\Large $(12,{\bf 3},4)$};
    \node[draw,fill=white,inner sep=2pt,minimum size=2pt] at (7.5,0) {\Large $(123,\emptyset,4)$};

    \node[draw,fill=white,inner sep=2pt,minimum size=2pt] at (2.5,4) {\Large $(\emptyset,12,34)$};
    \node[draw,fill=white,inner sep=2pt,minimum size=2pt] at (3.7,2) {\Large $(1,2,{\bf 34})$};
    \node at (3.7,2.7) {\Large $A$};
    
    \node[draw,fill=white,inner sep=2pt,minimum size=2pt] at (5,0) {\Large $(12,\emptyset,34)$};

    \node[draw,fill=white,inner sep=2pt,minimum size=2pt] at (1.3,2) {\Large $(\emptyset,1,234)$};
    \node[draw,fill=white,inner sep=2pt,minimum size=2pt] at (2.5,0) {\Large $(1,\emptyset,234)$};

    \node[draw,fill=white,inner sep=2pt,minimum size=2pt] at (0,0) {\Large $(\emptyset,\emptyset,1234)$};
 
\end{tikzpicture}
\caption{Illustration of a simplex whose vertices represent divisions of a path with four vertices into three bundles. In the small triangle colored in light gray, agent $A$ most-prefers the third bundle at the left corner vertex, $B$ most-prefers the first bundle at the middle corner vertex, and $C$ most-prefers the second bundle at the right corner vertex.}
\label{fig:Sperner}
\end{figure}%

\smallskip
\noindent
{\bf Related work}
To the best of our knowledge, the problem of dividing a discrete cake, i.e., a path, has been first considered in \citep{Marenco2014,Suksompong2019,BaranyGrinberg}. \citet{Marenco2014} studied a special valuation in which each item is liked by exactly one agent and showed that an envy-free connected division exists for such valuations. \citet{Suksompong2019} considered approximation of envy-freeness, showing that a simple rounding of an envy-free division gives us a connected division such that the envy is bounded by at most $2 v_{max}$ for agents with additive valuations, where $v_{max}$ is the maximum value of the agents for the single items; a result similar to this has been obtained in \citet{BaranyGrinberg}. 

The existence and complexity issues regarding other solution concepts on a discrete cake have been studied by a number of papers \citep{BouveretCeEl17,Misra,GoldbergHS20,Igarashi_Peters_2019,Bei_Igarashi_Lu_Suksompong_2021}. \citet{BouveretCeEl17} showed that deciding the existence of a connected division satisfying envy-freeness or proportionality is NP-hard; \citet{GoldbergHS20} strengthened these results by showing that the problems remain hard even when agents have binary additive valuations. 
\citet{Igarashi_Peters_2019} considered the relationship between fairness and efficiency. Unlike the standard fair division setting \citep{CaragiannisKuMo16}, it has been shown that EF1 and Pareto-optimality are incompatible under connectivity constraints of a path. They further showed that finding a Pareto-optimal connected division of a path satisfying approximate notions of fairness is NP-hard even when agents have binary additive valuations. 
\citet{Misra} focused on an approximate notion of equitability, called equitability up to one good (EQ1), and proved that an EQ1 connected division always exists and can be computed efficiently whenever agents have monotone valuations. 
\citet{Bei_Igarashi_Lu_Suksompong_2021} showed that a connected division satisfying a proportionality relaxation, called the \emph{invisible proportional share} (IPS) property, exists on a path for any number of agents with additive valuations; note that the IPS property is stronger than other proportionality relaxations, considered in \citep{AzizCaIg19,Suksompong2019,ConitzerFreemanShah}. 

\citet{BouveretCeEl17} proposed a model for allocating indivisible goods under connectivity constraints of a graph. The graph fair division has attracted a great deal of attention since then~\citep{Bouveret2019,Igarashi_Peters_2019,Bei_Igarashi_Lu_Suksompong_2021,TruszczynskiL20,BiloCaFl19,BouveretCeEl17,GrecoS20,parameterizedFair}. 
\citet{BiloCaFl19} developed several methods to obtain an EF1 division under connectivity constraints of a path when the number of agents is two, three, or four, or when the agents have identical monotone valuations. Further, they characterized the family of graphs for which an EF1 connected division always exists for two agents with monotone valuations.
\citet{parameterizedFair}, on the other hand, proved that the problem of deciding the existence of an EF1 connected division is NP-hard even when the graph is a star and the agents have binary additive valuations. 

Another concept of fairness that has been extensively studied in the study of graph fair division is the {\em maximin fair share criterion} (MMS), where the maximin fair share is defined for the set of all connected divisions. The problem of computing an MMS division is known to be polynomial-time solvable on trees \citep{BouveretCeEl17} and on cycles \citep{TruszczynskiL20}; however, it becomes intractable even when the underlying graph has a bounded treewidth \citep{GrecoS20}. \citet{Bei_Igarashi_Lu_Suksompong_2021} examined the gap between unconstrained MMS and graph-restricted MMS for various graphs. 


\section{Preliminaries}\label{sec:prem}
For each natural number $s \in \bbN$, we write $[s]=\{1,2,\ldots,s\}$. For each pair of natural numbers $s,t \in \bbN$ with $s \le t$, we write $[s,t]=\{s,s+1,\ldots,t\}$. 
We are given $n$ \emph{agents} and $m$ \emph{items} (or \emph{goods}). We may refer to subsets of items as {\em bundles}. The items are aligned along a path $(1,2,\ldots,m)$. Each agent $i$ has a \emph{valuation function} $v_i$ that assigns a real value to every connected subset of the path. For two connected subsets $S$ and $T$, an agent $i$ \emph{weakly prefers} (resp. \emph{strictly prefers}) $S$ to $T$ if $v_i(S) \geq v_i(T)$ (resp. $v_i(S) > v_i(T)$).
We assume that all agents have {\em monotone valuations}, i.e., each agent $i$ weakly prefers $T$ to $S$ whenever $S \subseteq T$ and that $v_i(\emptyset)=0$ for each $i \in N$. 

A \emph{division} is a partition $\calI=(I_1,I_2,\ldots,I_n)$ of the path into $n$ connected bundles, where $I_i$ is the $i$th bundle from the left. 
A division $\calI$ is \emph{envy-free} if there exists a permutation $\pi \colon [n] \rightarrow [n]$ such that $v_i(I_{\pi(i)}) \ge v_i(I_{\pi(j)})$ for any pair $i,j$ of agents. 

An envy-free division is not guaranteed when the items are indivisible. For instance, when two agents desire one item, one agent gets the item, whereas the other gets nothing. Thus, \citet{Budish11} relaxed the envy-freeness condition to {\em envy-freeness up to one good} (EF1). In an EF1 division, an agent can envy another agent, but envy will disappear after one item is removed from the other's bundle. We adopt a slightly more robust version of EF1, introduced in \citet{BiloCaFl19}, requiring that removing the items leave the envied bundles connected. 

\begin{definition}[EF1$_{outer}$: envy-freeness up to one \emph{outer} good]
A division $\calI$ satisfies \emph{EF1$_{outer}$} if there exists a permutation $\pi \colon [n] \rightarrow [n]$ such that for any pair $i,j$ of agents, $v_i(I_{\pi(i)}) \ge v_i(I_{\pi(j)})$, or there exists a good $g \in I_{\pi(j)}$ such that $I_{\pi(j)} \setminus \{g\}$ is connected\footnote{We consider the empty set to be connected.} and $v_i(I_{\pi(i)}) \ge v_i(I_{\pi(j)}\setminus \{g\})$.
\end{definition}

In our context, EF1$_{outer}$ is fairer than EF1. In particular, EF1 may not be binding at all when nonconnected subsets are less preferred compared with connected subsets or even undesirable; for instance, when allocating time slots for certain tasks among multiple employees, people often value being allocated a contiguous chuck of time, instead of being allocated to a disconnected one.  

We introduce the following notation in \citet{BiloCaFl19}. For every connected subset $I$, we define the \emph{up-to-one valuation} $v^-_i$ of agent $i$ as
\begin{align*}
	v^-_i(I) :=
	\begin{cases}
		0 & \text{if $I = \emptyset$,} \\
		\min \big\{v_i(I\setminus\{g\}) : \\
		g\in I \text{ such that } I\setminus\{g\}\text{ is connected} \big\} & \text{if $I \neq \emptyset$.}
	\end{cases}
\end{align*}
Clearly, a division $\calI$ satisfies EF1$_{outer}$ if and only if there exists a permutation $\pi \colon [n] \rightarrow [n]$ such that $v_i(I_{\pi(i)}) \ge v_i^-(I_{\pi(j)})$ for any pair $i,j$ of agents. 

Throughout this paper, we assume that the number of items exceeds the number of agents, i.e., $m\ge n$. Note that if $m < n$, a trivial EF1$_{outer}$ division exists.

\subsection{Sperner's lemma and the existence of an envy-free division}
We review basic notions of combinatorial topology and explain the link between Sperner's lemma and the classical cake cutting problem. 
An $(n-1)$-simplex $S$ is the convex hull of $n$ {\em main vertices} $\bfx_1,\bfx_2,\ldots,\bfx_n$; we write $S=\langle \bfx_1,\bfx_2,\ldots,\bfx_n \rangle$. For $j \in [n]$, we write $\bfe^{j} \in \{0,1\}^n$ as the $j$-th unit vector, where $e^j_h=1$ if $h=j$ and $e^j_h=0$ otherwise. The $(n-1)$ standard simplex $\Delta^{n-1}$ is the $(n-1)$-simplex whose main vertices are given by $\bfe^1,\bfe^2,\ldots,\bfe^n$. 
A {\em triangulation} $\ttT$ of an $(n-1)$-simplex $S$ is a collection of smaller $n-1$ simplices $S_1,S_2,\ldots,S_k$ where $S=\bigcup^k_{j=1}S_j$, and for each pair of distinct indices $i,j \in [k]$, the intersection $S_i \cap S_j$ is either empty or a face common to them. We refer to $S_1,S_2,\ldots,S_k$ as {\em elementary} simplices. We denote by $V(\ttT)$ the set of vertices of a triangulation $\ttT$.

\begin{figure*}[thb]
\centering
\begin{subfigure}[t]{.33\linewidth}
    \begin{tikzpicture}[scale=0.4, transform shape]
\begin{scope}[xshift=0cm,yshift=0cm]
     \draw [<->,thick] (-1,9) node (yaxis) [above] {\Large $x_2$}|- (9,-1) node (xaxis) [right] {\Large $x_1$};

    \draw[thick] (-1.2,0) node [left] {} -- (-0.8,0);
    \draw[thick] (-1.2,8) node [left] {} -- (-0.8,8);
    \draw[thick] (0,-1.2) node [below] {} -- (0,-0.8);
    \draw[thick] (8,-1.2) node [below] {} -- (8,-0.8);
    \draw[thick] (0,0) -- (0,8) -- (8,8) -- (0,0);
\end{scope}
\end{tikzpicture}
\subcaption{The standard simplex}
\end{subfigure}%
\begin{subfigure}[t]{.33\linewidth}
    \begin{tikzpicture}[scale=0.4, transform shape]
\begin{scope}[xshift=13cm,yshift=0cm]
     \draw [<->,thick] (-1,9) node (yaxis) [above] {\Large $x_2$}|- (9,-1) node (xaxis) [right] {\Large $x_1$};
    \draw[thick] (-1.2,0) node [left] {} -- (-0.8,0);
    \draw[thick] (-1.2,8) node [left] {} -- (-0.8,8);
    \draw[thick] (0,-1.2) node [below] {} -- (0,-0.8);
    \draw[thick] (8,-1.2) node [below] {} -- (8,-0.8);
    
    \draw[thick] (0,0) -- (0,8) -- (8,8) -- (0,0);
    \foreach \y in {2,4,6}{
            \draw[thick] (0,\y) rectangle (\y,\y);
    }
    \foreach \x in {2,4,6}{
            \draw[thick] (\x,8) rectangle (\x,\x);
    }    
    \foreach \x in {2,4,6}{
        \draw[thick] (\x, 8) -- (\x - 2, 6);
  }
    \foreach \x in {2,4}{
        \draw[thick] (\x, 6) -- (\x - 2, 4);
    }

    \draw[thick] (2, 4) -- (0, 2);

    \node[draw,circle,fill=white, inner sep=2pt,minimum size=2pt] at (0,8) {\Large B};
    \node[draw,circle,fill=white,inner sep=2pt,minimum size=2pt] at (2,8) {\Large C};
    \node[draw,circle,fill=white,inner sep=2pt,minimum size=2pt] at (4,8) {\Large A};
    \node[draw,circle,fill=white,inner sep=2pt,minimum size=2pt] at (6,8) {\Large B};
    \node[draw,circle,fill=white,inner sep=2pt,minimum size=2pt] at (8,8) {\Large C};

    \node[draw,circle,fill=white,inner sep=2pt,minimum size=2pt] at (0,6) {\Large A};
    \node[draw,circle,fill=white,inner sep=2pt,minimum size=2pt] at (2,6) {\Large B};
    \node[draw,circle,fill=white,inner sep=2pt,minimum size=2pt] at (4,6) {\Large C};
   \node[draw,circle,fill=white,inner sep=2pt,minimum size=2pt] at (6,6) {\Large A};

    \node[draw,circle,fill=white,inner sep=2pt,minimum size=2pt] at (0,4) {\Large C};
    \node[draw,circle,fill=white,inner sep=2pt,minimum size=2pt] at (2,4) {\Large A};
    \node[draw,circle,fill=white,inner sep=2pt,minimum size=2pt] at (4,4) {\Large B};

    \node[draw,circle,fill=white,inner sep=2pt,minimum size=2pt] at (0,2) {\Large B};
    \node[draw,circle,fill=white,inner sep=2pt,minimum size=2pt] at (2,2) {\Large C};
    \node[draw,circle,fill=white,inner sep=2pt,minimum size=2pt] at (0,0) {\Large A};
\end{scope}
\end{tikzpicture}
\subcaption{Owner labeling}
\end{subfigure}%
\begin{subfigure}[t]{.33\linewidth}
    \begin{tikzpicture}[scale=0.4, transform shape]
\begin{scope}[xshift=26cm,yshift=0cm]
     \draw [<->,thick] (-1,9) node (yaxis) [above] {\Large $x_2$}|- (9,-1) node (xaxis) [right] {\Large $x_1$};

    \filldraw[gray!10] (0,4) -- (0,6) -- (2,6) -- (0,4);
    \filldraw[gray!10] (2,4) -- (2,6) -- (4,6) -- (2,4);
    \filldraw[gray!10] (4,6) -- (6,6) -- (6,8) -- (4,6);
    \filldraw[gray!10] (2,4) -- (4,4) -- (4,6) -- (2,4);
    \filldraw[gray!10] (4,4) -- (4,6) -- (6,6) -- (4,4);


    \draw[thick] (-1.2,0) node [left] {} -- (-0.8,0);
    \draw[thick] (-1.2,8) node [left] {} -- (-0.8,8);
    \draw[thick] (0,-1.2) node [below] {} -- (0,-0.8);
    \draw[thick] (8,-1.2) node [below] {} -- (8,-0.8);

    \draw[thick] (0,0) -- (0,8) -- (8,8) -- (0,0);

    \foreach \y in {2,4,6}{
            \draw[thick] (0,\y) rectangle (\y,\y);
    }
    \foreach \x in {2,4,6}{
            \draw[thick] (\x,8) rectangle (\x,\x);
    }    

    \foreach \x in {2,4,6}{
        \draw[thick] (\x, 8) -- (\x - 2, 6);
    }

    \foreach \x in {2,4}{
        \draw[thick] (\x, 6) -- (\x - 2, 4);
    }

    \draw[thick] (2, 4) -- (0, 2);

    \node[draw,circle,fill=red!30, inner sep=2pt,minimum size=2pt] at (0,8) {\Large B};
    \node[draw,circle,fill=red!30,inner sep=2pt,minimum size=2pt] at (2,8) {\Large C};
    \node[draw,circle,fill=red!30,inner sep=2pt,minimum size=2pt] at (4,8) {\Large A};
    \node[draw,circle,fill=blue!30,inner sep=2pt,minimum size=2pt] at (6,8) {\Large B};
    \node[draw,circle,fill=blue!30,inner sep=2pt,minimum size=2pt] at (8,8) {\Large C};

    \node[draw,circle,fill=red!30,inner sep=2pt,minimum size=2pt] at (0,6) {\Large A};
    \node[draw,circle,fill=blue!30,inner sep=2pt,minimum size=2pt] at (2,6) {\Large B};
    \node[draw,circle,fill=red!30,inner sep=2pt,minimum size=2pt] at (4,6) {\Large C};
    \node[draw,circle,fill=white,inner sep=2pt,minimum size=2pt] at (6,6) {\Large A};

    \node[draw,circle,fill=white,inner sep=2pt,minimum size=2pt] at (0,4) {\Large C};
    \node[draw,circle,fill=white,inner sep=2pt,minimum size=2pt] at (2,4) {\Large A};
    \node[draw,circle,fill=blue!30,inner sep=2pt,minimum size=2pt] at (4,4) {\Large B};

    \node[draw,circle,fill=red!30,inner sep=2pt,minimum size=2pt] at (0,2) {\Large B};
    \node[draw,circle,fill=white,inner sep=2pt,minimum size=2pt] at (2,2) {\Large C};

    \node[draw,circle,fill=white,inner sep=2pt,minimum size=2pt] at (0,0) {\Large A};
 
\end{scope}
\end{tikzpicture}
\subcaption{Coloring}
\end{subfigure}
\caption{Illustration of the labeling-coloring approach when there are three agents $A,B,C$. 
The blue, red, and white vertices in Figure (c) correspond to the divisions where the first, second, and third pieces are the most favorite pieces of an owner agent, respectively. 
The three elementary simplices colored in light gray in Figure (c) correspond to fully-colored simplices.
}
\label{fig:Triangulation}
\end{figure*}
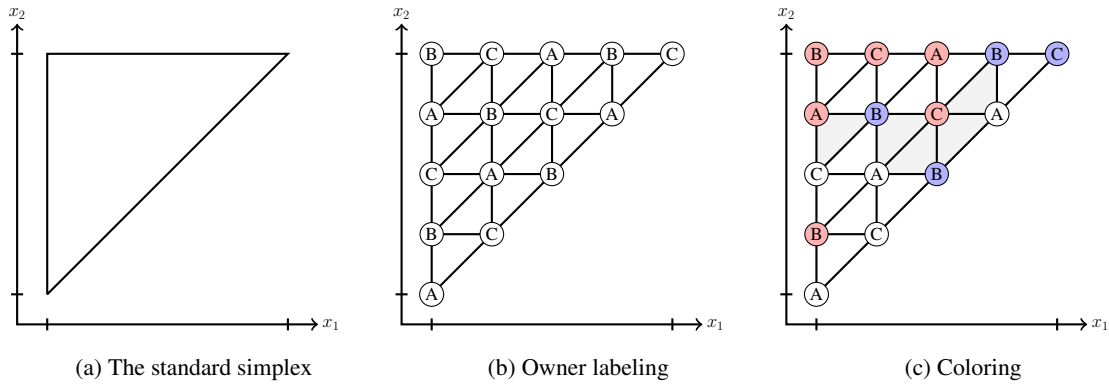%

Given a triangulation $\ttT$ of an $(n-1)$-simplex $S$, a {\em coloring} is a function $\lambda \colon V(\ttT) \rightarrow 2^{[n]}$ that assigns to each vertex $\bfx \in V(\ttT)$ a subset $\lambda(\bfx) \subseteq [n]$, where each element of $[n]$ is (called) a {\em color}. A coloring $\lambda \colon V(\ttT) \rightarrow 2^{[n]}$ is (called) {\em proper} if we can write $S=\langle \bfx_1,\bfx_2,\ldots,\bfx_n \rangle$ such that if a vertex $\bfx \in V(\ttT)$ is colored by index $j$, i.e., $j \in \lambda (\bfx)$, then $\bfx_j$ is a vertex of a minimal face containing $\bfx$. 

A {\em fully-colored} elementary simplex $S^*=\langle \bfx^*_1,\bfx^*_2,\ldots,\bfx^*_n \rangle$ has a complete set of colors, i.e., there exists a permutation $\pi \colon [n]\rightarrow [n]$ such that $\pi(i) \in \lambda(\bfx^*_i)$ for each $i \in [n]$. Sperner's lemma states that a triangulated simplex with a proper coloring admits a fully-colored elementary simplex.

\begin{theorem}[Sperner's lemma]\label{thm:sperner}
Any triangulation $\ttT$ of an $(n-1)$-simplex with a proper coloring $\lambda \colon V(\ttT) \rightarrow 2^{[n]}$ admits a fully-colored elementary simplex.
\end{theorem}

\citet{Su1999} was the first to demonstrate the usefulness of Sperner's lemma in the context of cake-cutting, citing Simmons as the one who conceptualized the proof. The Simmons--Su method encoded possible divisions of the cake $[0,1]$ by the points of an $(n-1)$ standard simplex where the $i$-th coordinate can be interpreted as the $i$-th knife position, obtaining an envy-free division based on a {\em labeling} and {\em coloring} of the $(n-1)$ standard simplex as follows.

First, 
we assign an {\em agent label} to each vertex so that the vertices of each elementary simplex have $n$ distinct agent labels. 
Formally, an {\em owner labeling} of the triangulation $\ttT$ of an $(n-1)$-simplex is a function $a \colon V(\ttT) \rightarrow [n]$ such that for each pair of distinct vertices $\bfx_i$ and $\bfx_j$ of $S$ with $i \neq j$, $a(\bfx_i) \neq a(\bfx_j)$, where each $a(\bfx)$ is called the {\em owner} of a vertex. There is a triangulation of the simplex that does not admit an owner labeling; see, e.g., Figure $1$ in \citet{Deng2012}. Nevertheless, some triangulations, e.g., Kuhn's triangulation and the barycentric triangulation, do admit owner labelings, while allowing small mesh size~\citep{Su1999,Deng2012}.

Second, since the vertices of the triangulation correspond to divisions of the cake, we go to each vertex of the triangulation and ask the owner to {\em color} that vertex with the indices of the most preferred piece of the corresponding division. An illustration of a triangulation and the labeling-coloring approach is given in Figure~\ref{fig:Triangulation}. Such a coloring is proper if each agent never chooses the piece of zero-length. For example, when $n=3$, the corner vertices have distinct colors as they correspond to divisions in which one piece includes the entire cake; further, the vertices on each side correspond to divisions in which one piece is empty, thus missing one color that corresponds to the empty piece. By Sperner’s lemma, this construction yields a fully-colored elementary simplex $S^*$. Now, observe that in $S^*$, each agent points out a different piece of her owned division as a most preferred piece. If the simplex vertices are close enough, it corresponds to an approximate solution that converges to an envy-free contiguous division for a sequence of finer triangulations. In the next section, we will adopt the Simmons--Su method to the setting of discrete cake-cutting.

\section{Existence of EF1 for any number of agents}\label{sec:EF1}
In this section, we prove the main result of this paper. 

\begin{theorem}\label{thm:EF1:general}
For any number of agents with monotone valuation functions on a path, a connected EF1$_{\emph outer}$ division exists.
\end{theorem}

Let us first illustrate potential approaches to prove the existence of a connected EF1$_\emph{outer}$ division. We will then proceed to identifying some of the problems with these approaches, and explain how we overcome them.

\subsection{Potential approach to prove Theorem~\ref{thm:EF1:general}}\label{sec:potential}

One possible approach to prove Theorem~\ref{thm:EF1:general} is to treat a path $(1,2,\ldots,m)$ as an interval $[0,m]$, extend valuation functions to continuous ones, and apply some rounding of an envy-free connected division of the cake, which is guaranteed to exist~\citep{Su1999,Woodall1980,Stromquist80}. This approach was also investigated by \citet{Suksompong2019}; however, a division obtained via this approach satisfies a weaker fairness property. In particular, it can lead agents to {steal} an item instead of removing it from another agent's bundle to eliminate envy; see Section $5$ of \citet{BiloCaFl19} for details.

Instead, \citet{BiloCaFl19} followed the approach of \citet{Su1999} that directly works on the simplex and rounds an elementary simplex, corresponding to a series of divisions. Like \citet{Su1999}, they encoded possible configurations of the $n-1$ knives as the vertices of a triangulated simplex. 
Bil\`{o} et al. observed that if the configuration space only considers fully integral divisions, i.e., if knives move from one edge to another edge, the divisions corresponding to the vertices in each elementary simplex are too far apart from each other to ensure EF1: in their final rounding, one agent may get one additional item together with her desired bundle while another may lose one item, which makes it difficult to bound the envy up to one item. Thus, they consider a finer triangulation, allowing knives to move at both edges and vertices of a path. More precisely, they consider the following simplex: 
\[
S_m:= \left\{ \bfx \in \bbR^{n-1}_{+} \middle| \frac{1}{2} \le x^1 \le x^2 \le \cdots \le x^{n-1} \le m +\frac{1}{2}  \right\}, 
\]
and use a {\em Kuhn's triangulation} $\ttT_{\half}$ of $S_m$ \citep{Deng2012} where the vertices $V(\ttT_{\half})$ are given by 
\[
V(\ttT_{\half})= \left\{  \bfx \in \bbR^{n-1}_{+}  \middle| x^i  \in  \{\frac{1}{2},1, \ldots, m +\frac{1}{2} \}, \forall i \in [n]  \right\}
\]
and it satisfies the property that each elementary simplex $S=\langle \bfx_1,\bfx_2,\ldots,\bfx_n \rangle$ of $\ttT_{\half}$ is balanced, meaning that there exists a permutation $\phi \colon [n] \rightarrow [n]$ such that 
\begin{align}\label{eq:sperner:kuhn-triangulation}
\bfx_{\phi(i+1)}=\bfx_{\phi(i)}+\frac{1}{2} \bfe^{\phi(i)},~\mbox{for each}~i \in [n-1]. 
\end{align}

The above property ensures that the vertices of the elementary simplex can be arranged in such a way that every distinct knife in this sequence moves always in half-step and in the same direction and such a movement happens exactly once over the sequence. 

For each vertex $\bfx=(x^1,x^2,\ldots,x^{n-1}) \in V(\ttT_{\half})$, we call each $x^j$ the $j$-th {\em knife position}; item $y \in \{ 1, 2, \dots, m \}$ is {\em hidden} by a knife $x^j$ if $y=x^j$. Here, $V(\ttT_{\half})$ 
encodes all the possible configurations $(x^1,x^2,\ldots,x^{n-1})$ of the $n-1$ knives that move in half-steps. Specifically, each vertex $\bfx \in V(\ttT_{\half})$ induces a partial division $\calI({\bfx})=(I_1(\bfx),I_2(\bfx),\ldots,I_n(\bfx))$ where each $j$-th bundle for $j \in [n]$ is given by
\[
I_j(\bfx) = \,\{ y \in \{ 1, 2, \dots, m \}  \mid x^{j-1} < y < x^j \,\}, 
\]
setting $x^0 = \frac12$ and $x^n = m+\frac12$. For example, when $m=12$, $n=5$, and $\bfx=(3,\frac{9}{2},8,\frac{21}{2})$, $\calI({\bfx})$ corresponds to the partial division $(\{1,2\},\{4\},\{5,6,7\},\{9,10\},\{11,12\})$. Items $3$ and $8$ are hidden by the two knives $x^2$ and $x^3$ and the other items are uncovered at $\bfx$. For each $\bfx \in V(\ttT_{\half})$ and each $j \in [n]$, we say that $\ell_{j}(\bfx) \coloneqq \floor{x^{j-1}+\frac{1}{2}}$ is the {\em left-most boundary item} of $I_j(\bfx)$; similarly, we say that $r_{j}(\bfx) \coloneqq \ceil{x^j-\frac{1}{2}}$ is the {\em right-most boundary item} of $I_j(\bfx)$. Note that if $\ell_{j}(\bfx)$ is hidden by the $(j-1)$-th knife, it is the boundary item between $I_{j-1}(\bfx)$ and $I_j(\bfx)$ and if not, it is the leftmost item of $I_j(\bfx)$; similarly, if $r_{j}(\bfx)$ is hidden by the $j$-th knife, it is the boundary item between $I_{j}(\bfx)$ and $I_{j+1}(\bfx)$, and if not, it is the rightmost item of $I_j(\bfx)$. 
We say that item $y \in \{ 1, 2, \dots, m \}$ {\em fully appears} (or, is {\em fully visible}) in $I_j(\bfx)$ if $y \in I_j(\bfx)$.

How does each agent evaluate each of the partial divisions $\calI({\bfx})$? Bil\`{o} et al. introduced the following \emph{virtual} valuations: for a bundle $I_j({\bfx})$ in which at least one of the boundary items is fully visible, agents expect to obtain only those items that are fully visible in the bundle; for other bundles $I_j({\bfx})$ in which no boundary item fully appears, agents expect to obtain the items that are fully visible in the bundle as well as at least one of the boundary items (choose one that is less valuable). Agents then color each of the vertices with the index of the favorite bundles based on virtual valuations. This coloring can be shown to be proper; thus, using a more general version of Sperner's lemma considering $n$ colorings, we get a fully-colored simplex $S^*$ that is fully-labeled, meaning that $S^*$ receives different agent labels that like different bundles best.\footnote{Because Kuhn's triangulation admits an owner labeling, instead of using a general version of Sperner's lemma, the labeling-coloring approach of \citet{Su1999} also shows the existence of $S^*$. In our proof of the next subsection, we use this approach to construct a proper coloring as the proof becomes slightly more elementary.}

Note that each vertex in $S^*$ only induces a partial division. How can we use $S^*$ to obtain a full division? 
Bil\`{o} et al. showed that $S^*$ can be rounded to yield an EF1$_{\emph outer}$ connected division for four or fewer agents as follows. 
First of all, recall that for each elementary simplex $S=\langle \bfx_1,\bfx_2,\ldots,\bfx_n \rangle$ of $\ttT_{\half}$, there exists a permutation $\phi:[n] \rightarrow [n]$ satisfying \eqref{eq:sperner:kuhn-triangulation}. 
Thus, $S$ corresponds to a sequence of partial divisions $(\calI({\bfx_{\phi(1)}}),\calI({\bfx_{\phi(2)}}),\ldots,\calI({\bfx_{\phi(n)}}))$, where each partial division $\calI({\bfx_{\phi(k)}})$ is obtained from $\calI({\bfx_{\phi(k-1)}})$ by moving one of the knives in half-step and the movement of each knife occurs only once across $n$ partial divisions. An example of such sequences of partial divisions is shown in Figure~\ref{fig:Ix}. Using this, our path can be divided into $(B_1,y^1,B_2,\ldots,y^{n-1},B_n)$ where each $B_j$ is the set of items that fully appear in the $j$-th bundle of all partial divisions, and each $y^j$ is the boundary item that can appear in both of the $j$-th and the $(j+1)$-th bundles over the sequence.

Building $(B_1,y^1,B_2,\ldots,y^{n-1},B_n)$ for $S^*$, a division $\mathcal{I}^*=(I^*_1,I^*_2,I^*_3,I^*_4)$ for four agents can be constructed as follows: 
\begin{enumerate}
    \item Each $I^*_j$ of interior bundles with $j \in \{2,3\}$ consists of $B_j$ together with the boundary items $y^{j-1}$ or $y^j$ fully appearing in the $j$-th bundle of some partial division represented by the vertices in $S^*$. If none of the boundary items $y^{j-1}$ or $y^j$ fully appears in the $j$-th bundle, $I^*_j$ additionally gets one boundary item that is adjacent to an exterior bundle $I^*_j$ with $j \in \{1,4\}$. 
    \item Each exterior bundle $I^*_1$ (respectively, $I^*_4$) consists of $B_1$ together with $y_1$ (respectively, $y_3$) if $y_1$ (respectively, $y_3$) is not allocated yet. 
\end{enumerate}
With this approach, Bil\`{o} et al. established an EF1$_{\emph outer}$ connected division for four or fewer agents. However, the four-agent proof in Bil\`{o} et al. does not extend to the general case. Their proof requires that each interior bundle receives an item traversed by a knife due to the left-right symmetry in their valuations. This requirement is met by four agents as the interior bundles (i.e., the second and third bundles) are adjacent to some exterior bundle (i.e., the first and fourth bundles). When the number of agents is five or more, any division includes a bundle that is not adjacent to any of exterior bundles.  

\begin{figure*}[htb]
\centering
\begin{tikzpicture}[xscale=1]

\draw[decorate,decoration = {calligraphic brace}, thick] (0.8,1) -- node[above=1ex]{$B_1$} (2.2,1);

\draw[decorate,decoration = {calligraphic brace}, thick] (3.7,1) -- node[above=1ex]{$B_2$} (4.2,1);

\draw[decorate,decoration = {calligraphic brace}, thick] (5.7,1) -- node[above=1ex]{$B_3$} (7.2,1);

\draw[decorate,decoration = {calligraphic brace}, thick] (8.7,1) -- node[above=1ex]{$B_4$} (10.2,1);

\draw[decorate,decoration = {calligraphic brace}, thick] (11.7,1) -- node[above=1ex]{$B_5$} (12.2,1);

\node at (3,1.5) {$y^1$};
\node at (5,1.5) {$y^2$};
\node at (8,1.5) {$y^3$};
\node at (11,1.5) {$y^4$};

\begin{scope}[yshift=0cm]
\draw[-latex,thick,gray] (4.5,0) -- (4.7,-0.4);
\draw[thick,-] (1,0.5) node[anchor=east] {$\calI(\bfx_{\phi(1)})$~~~~~~~~}  -- (12,0.5);
\foreach \x in {1,2,3,4,5,6,7,8,9,10,11,12}
   \node[fill=white,draw, circle,inner sep=3pt,font=\tiny] (o\x) at (\x, 0.5){$\x$};
   
\draw[ultra thick, gray] (3,0.8) node[above] {$x^1$} -- (3,0.2);
\draw[ultra thick, gray] (4.5,0.8) node[above] {$x^2$} -- (4.5,0.2);
\draw[ultra thick,gray] (8,0.8) node[above] {$x^3$} -- (8,0.2);
\draw[ultra thick,gray] (10.5,0.8) node[above] {$x^4$} -- (10.5,0.2);
\end{scope}

\begin{scope}[yshift=-1.5cm]

\draw[-latex,thick,gray] (8,0) -- (8.2,-0.4);
\draw[thick,-] (1,0.5) node[anchor=east] {$\calI(\bfx_{\phi(2)})$~~~~~~~~}  -- (12,0.5);
\foreach \x in {1,2,3,4,5,6,7,8,9,10,11,12}
   \node[fill=white,draw, circle,inner sep=3pt,font=\tiny] (o\x) at (\x, 0.5){$\x$};

\draw[ultra thick, gray] (3,0.8) node[above] {$x^1$} -- (3,0.2);
\draw[ultra thick, gray] (5,0.8) node[above] {$x^2$} -- (5,0.2);
\draw[ultra thick,gray] (8,0.8) node[above] {$x^3$} -- (8,0.2);
\draw[ultra thick,gray] (10.5,0.8) node[above] {$x^4$} -- (10.5,0.2);
\end{scope}

\begin{scope}[yshift=-3cm]
\draw[-latex,thick,gray] (10.5,0) -- (10.7,-0.4);

\draw[thick,-] (1,0.5) node[anchor=east] {$\calI(\bfx_{\phi(3)})$~~~~~~~~} -- (12,0.5);
\foreach \x in {1,2,3,4,5,6,7,8,9,10,11,12}
   \node[fill=white,draw, circle,inner sep=3pt,font=\tiny] (o\x) at (\x, 0.5){$\x$};

\draw[ultra thick, gray] (3,0.8) node[above] {$x^1$} -- (3,0.2);
\draw[ultra thick, gray] (5,0.8) node[above] {$x^2$} -- (5,0.2);
\draw[ultra thick,gray] (8.5,0.8) node[above] {$x^3$} -- (8.5,0.2);
\draw[ultra thick,gray] (10.5,0.8) node[above] {$x^4$} -- (10.5,0.2);
\end{scope}

\begin{scope}[yshift=-4.5cm]

\draw[-latex,thick,gray] (3,0) -- (3.2,-0.4);
\draw[thick,-] (1,0.5) node[anchor=east] {$\calI(\bfx_{\phi(4)})$~~~~~~~~} -- (12,0.5);
\foreach \x in {1,2,3,4,5,6,7,8,9,10,11,12}
   \node[fill=white,draw, circle,inner sep=3pt,font=\tiny] (o\x) at (\x, 0.5){$\x$};

\draw[ultra thick, gray] (3,0.8) node[above] {$x^1$} -- (3,0.2);
\draw[ultra thick, gray] (5,0.8) node[above] {$x^2$} -- (5,0.2);
\draw[ultra thick,gray] (8.5,0.8) node[above] {$x^3$} -- (8.5,0.2);
\draw[ultra thick,gray] (11,0.8) node[above] {$x^4$} -- (11,0.2);
\end{scope}

\begin{scope}[yshift=-6cm]
\draw[thick,->] (0,0) -- (13,0) node[anchor=north west] {$x^i$};

\draw[thick,-] (1,0.5) node[anchor=east] {$\calI(\bfx_{\phi(5)})$~~~~~~~~}  -- (12,0.5);
\foreach \x in {1,2,3,4,5,6,7,8,9,10,11,12}
   \node[fill=white,draw, circle,inner sep=3pt,font=\tiny] (o\x) at (\x, 0.5){$\x$};

\foreach \x in {1,2,3,4,5,6,7,8,9,10,11,12}
   \draw (\x cm,1pt) -- (\x cm,-1pt) node[anchor=north,font=\tiny] {$\x$};

\draw[ultra thick, gray] (3.5,0.8) node[above] {$x^1$} -- (3.5,0.2);
\draw[ultra thick, gray] (5,0.8) node[above] {$x^2$} -- (5,0.2);
\draw[ultra thick,gray] (8.5,0.8) node[above] {$x^3$} -- (8.5,0.2);
\draw[ultra thick,gray] (11,0.8) node[above] {$x^4$} -- (11,0.2);
\end{scope}

\end{tikzpicture}
\caption{Example of a sequence of partial divisions represented by an elementary simplex of $\ttT_{\half}$. Each knife moves in half-step and the movement happens only once over the sequence.}\label{fig:Ix}
\end{figure*}
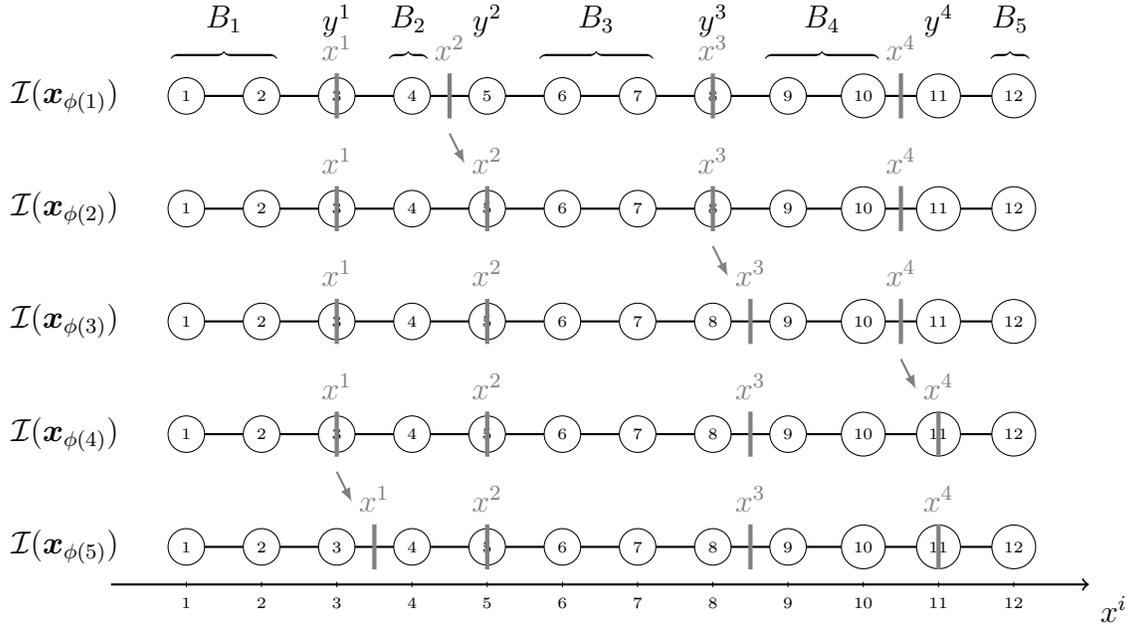%
\begin{figure*}[htb]
\centering
\begin{tikzpicture}[xscale=1]
\node at (0,0) {};
\begin{scope}[xshift=1cm,yshift=0cm]
\draw[thick,-] (1,0.5) node[anchor=east] {$\calI^*$~~~~~~~~}  -- (12,0.5);
\foreach \x in {1,2,3,4,5,6,7,8,9,10,11,12}
\node[fill=white,draw, circle,inner sep=3pt,font=\tiny] (o\x) at (\x, 0.5){$\x$};

\draw [draw=black] (0.5,0.2) rectangle (3.5,0.8);
\draw [draw=black] (3.5,0.2) rectangle (5.5,0.8);
\draw [draw=black] (5.5,0.2) rectangle (8.5,0.8);
\draw [draw=black] (8.5,0.2) rectangle (10.5,0.8);
\draw [draw=black] (10.5,0.2) rectangle (12.5,0.8);

\foreach \x in {1,2,3,4,5,6,7,8,9,10,11,12}
   \draw (\x cm,1pt) -- (\x cm,-1pt) node[anchor=north,font=\tiny] {$\x$};
\draw[thick,->] (0,0) -- (13,0) node[anchor=north west] {$x^i$};
\end{scope}
\end{tikzpicture}
\caption{Division returned by Algorithm~\ref{alg:rounding} for an elementary simplex whose vertices correspond to the partial divisions of Figure~\ref{fig:Ix}. Note that the second bundle $I^*_2$ receives the boundary item $5$ because it satisfies the condition $($b$)$ of Line $9$ while the fourth bundle $I^*_4$ does not receive item $11$ because it violates the condition $($b$)$ of Line $9$.}\label{fig:IxFull}
\end{figure*}
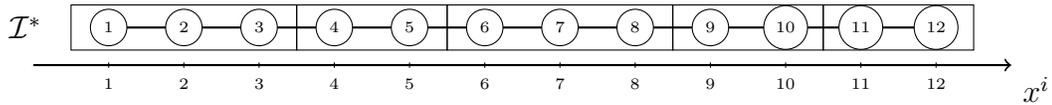

\subsection{Our approach: Proof of Theorem \ref{thm:EF1:general}}
Following the discretization approach developed in \citet{BiloCaFl19}, we prove that an EF1 connected division exists for any number of agents with monotone valuations. When extending the result beyond four agents, the main difficulty is to find appropriate ways to evaluate bundles of partial divisions and to round the half-integral simplices. To this end, we create a left-right asymmetry both in the evaluation phase and in partial-division rounding. 
Intuitively, our rounding algorithm, formalized in Algorithm~\ref{alg:rounding}, prioritizes a left agent (i.e., an agent who is assigned to $j$-th bundle) over a right agent (i.e., an agent who is assigned to $(j+1)$-th bundle) when allocating each boundary item $y^j$. 
Our definition of virtual valuations allows us to do such rounding since each agent is pessimistic about obtaining the left-most boundary item and optimistic about obtaining the right-most boundary item. This then guarantees that the value of each agent's allocated bundle in the final division $\mathcal{I}^*$ is at least the value of the favorite bundle in their owned division, which is at least the value of another bundle in $\mathcal{I}^*$ after the removal of one outer-item. Below, we prove that our technique successfully ensures that the final division is EF1$_{\emph outer}$.

Our proof is divided into the following two steps of coloring and rounding: first, we assign to each vertex an owner labeling and a color according to the preferences of owners; second, we round a fully-colored simplex into a full division.
See Figure~\ref{fig:Triangulation} for an illustration. 
We use the same simplex $S_m$ and triangulation $V(\ttT_{\half})$ as previously described. 

\subsubsection{Coloring}
We define the virtual valuation $\hat v_i(\bfx, j)$ of each vertex $\bfx$ of the triangulation. The virtual valuation determines how each owner agent assigns a color to his or her owned vertex. 
First, an agent who obtains the left exterior bundle $j=1$ expects to obtain the items not hidden by the right-most knife:
\begin{align*}
\hat v_i(\bfx, 1) &=v_i(I_1(\bfx)). 
\end{align*}

For the interior bundles $2\le j \le n-1$, we set ${\hat v}_i(\bfx, j)=0$ if $I_j(\bfx) = \emptyset$. Otherwise, 
the value ${\hat v}_i(\bfx, j)$ is given as follows: 
\begin{align*}
\hat v_i(\bfx, j) &=
\begin{cases}
v_i^-(I_j(\bfx)\cup \{\ell_{j}(\bfx), r_{j}(\bfx)\}) &
\text{if $\ell_{j}(\bfx) \not \in I_j(\bfx)$ and $r_{j}(\bfx) \not \in I_j(\bfx)$},\\
v_i(I_j(\bfx) \setminus \{\ell_{j}(\bfx)\}) &
\text{if $\ell_{j}(\bfx) \in I_j(\bfx)$ and $r_{j}(\bfx) \in I_j(\bfx)$},\\
v_i^-(I_j(\bfx) \cup \{r_{j}(\bfx)\}) &
\text{if $\ell_{j}(\bfx)\in I_j(\bfx)$ and $r_{j}(\bfx) \not \in I_j(\bfx)$},\\
v_i(I_j(\bfx)) &
\text{if $\ell_{j}(\bfx) \not \in I_j(\bfx)$ and $r_{j}(\bfx) \in I_j(\bfx)$}.\\ 
\end{cases}
\end{align*}

For the right exterior bundle $j=n$, the item $\ell_n(\bfx)$ is not expected in the final bundle $I^*_j$: 
\begin{align*}
\hat v_i(\bfx, n) &=v_i(I_n(\bfx)\setminus \{\ell_n(\bfx)\})
\end{align*}

Based on these virtual valuations, we define coloring functions ${\lambda}_i \colon V(\ttT_\half) \to 2^{[n]}$ for each agent $i \in [n]$ where 
\[ 
{\lambda}_i(\bfx) = \argmax \{\, {\hat v}_i(\bfx, j) \mid j \in [n]~\mbox{such that}~I_j(\bfx) \neq \emptyset \,\}. 
\]
It is not difficult to see that these colorings are proper. It is known that Kuhn's triangulation admits an owner labeling $a \colon V(\ttT_{\half}) \rightarrow [n]$ where each elementary simplex has distinct owner labels~\citep{Deng2012}. 
We aggregate the colorings according to the preference of each owner as follows: 
\[
{\lambda}(\bfx)= {\lambda}_{a(\bfx)}(\bfx). 
\]
Because each coloring ${\lambda}_i$ is proper, so is ${\lambda}$. 

\subsubsection{Rounding}
Next, we present how to round each elementary simplex $S=\langle \bfx_1,\bfx_2,\ldots,\bfx_n \rangle$ of $\ttT_{\half}$ into a division $\mathcal{I}^*=(I^*_1,I^*_2,\ldots,I^*_n)$ as specified in Algorithm \ref{alg:rounding}. Initially, in Line~\ref{line:initialization}, each $I^*_j$ is allocated to the set $B_j$ of items that fully appear in all the $j$-th bundles $I_j(\bfx_1),\ldots,I_j(\bfx_n)$ of partial partitions represented by the elementary simplex $S$. In Lines \ref{line:boundary1} -- \ref{line:boundary2}, Algorithm \ref{alg:rounding} allocates, from left to right, each boundary item $y^j$ that appears in both $j$th and $j+1$st bundles on the sequence $I_j(\bfx_1),\ldots,I_j(\bfx_n)$. Specifically, we use the following left-right asymmetric rounding (Line~\ref{line:boundary:condition}): each bundle $I^*_j$ obtains item $y^{j}$ if 
\begin{itemize}
    \item[(a)] $y^{j}$ fully appears in some of the $j$-th bundles $I_j(\bfx_1),\ldots,I_j(\bfx_n)$, or
    \item[(b)] $I^*_j$ does not obtain item $y^{j-1}$ in the previous step and none of the $I_j(\bfx_k)$ coincides with $B_j$ for $k \in [n]$.
\end{itemize}
In this way, each interior bundle $I^*_j$ must receive at least one of its adjacent boundary items, except when 
$I_j(\bfx_k)$ coincides with $B_j$ for some $k \in [n]$. See Figure~\ref{fig:IxFull} for an example of the division returned by Algorithm~\ref{alg:rounding}.


\begin{algorithm}[ht]
	\caption{Rounding into a division}
	\label{alg:rounding}
	\begin{algorithmic}[1]
		\REQUIRE an elementary simplex $S=\langle \bfx_1,\bfx_2,\ldots,\bfx_n \rangle$ of $\ttT_{\half}$
		\ENSURE a division $\calI^*=(I^*_1,I^*_2,\ldots,I^*_n)$ of the path $(1,2,\ldots,m)$
		\STATE for each $j \in [n]$, let $B_j=I_j(\bfx_1) \cap I_j(\bfx_2) \cap \cdots \cap I_j(\bfx_n)$ and let $y^j$ be the item such that $y^j=x^j_k$ for some $k \in [n]$
		\STATE initialize $I^*_j \leftarrow B_j$ for each $j=1,2,\ldots,n$\label{line:initialization}
		\IF{$y^1$ fully appears in the first bundle of some division, i.e., $y^1 \in I_1(\bfx_k)$ for some $k$}\label{line:boundary1}
			\STATE $I^*_1 \leftarrow I^*_1 \cup \{y^1\}$
		\ENDIF
		\FOR{$j = 2,\dots,n-1$}
			\IF{$y^{j-1}$ is unallocated, i.e., $y^{j-1} \not \in \bigcup_{k \le j-1} I^*_{k}$}
				\STATE $I^*_j \leftarrow I^*_j \cup \{y^{j-1}\}$
			\ENDIF
			\IF{$y^{j-1} \neq y^j$}
			\IF{$(a)$ $y^j$ fully appears in the $j$-th bundle of some partial division, i.e., $y^j \in I_j(\bfx_k)$ for some $k \in [n]$; or $(b)$ $y^{j-1}$ is already allocated to some other agent, i.e., $y^{j-1} \in \bigcup_{k \le j-1} I^*_{k}$, and there is no $k \in [n]$ with $x^{j-1}_{k}=y^{j-1}+\frac 12$ and $x^{j}_{k}= y^{j}-\frac 12$\label{line:boundary:condition}
			}  
			\label{line:if}
				\STATE $I^*_j \leftarrow I^*_j \cup \{y^{j}\}$\label{line:boundary2}
			\ENDIF
			\ENDIF
		\ENDFOR
		\IF{$y^{n-1}$ is unallocated, i.e., $y^{n-1} \not \in \bigcup_{k \le n-1} I^*_{k}$}
			\STATE $I^*_{n} \leftarrow I^*_n \cup \{y^{n-1}\}$
		\ENDIF	
		\RETURN $(I^*_1,I^*_2,\ldots,I^*_n)$
	\end{algorithmic}
\end{algorithm}

In Lemma~\ref{lem:ef1:approx-correct}, we show that the estimate of each agent $i$ on the $j$-th bundle of the output is neither overly optimistic nor overly pessimistic when applying Algorithm \ref{alg:rounding} to any elementary simplex $S$ of $\ttT_{\half}$. The case distinction in the proof of Lemma~\ref{lem:ef1:approx-correct} considers whether the bundle is an interior or exterior one, whether the bundle receives two boundary items, exactly one boundary item, or none of them, and whether each of the knives that induce the bundle is located left to the boundary item, at the boundary item, or right to the boundary item.

\begin{lemma}
	\label{lem:ef1:approx-correct}
	Consider the triangulation $\ttT_{\half}$ of $S_m$. Let $S=\langle \bfx_1,\bfx_2,\ldots,\bfx_n \rangle$ be any elementary simplex of $\ttT_{\half}$ and let $(I^*_1,I^*_2,\ldots,I^*_n)$ be a division returned by Algorithm \ref{alg:rounding}. Then, for each $i, j, k\in [n]$, we have
	$v_i(I^*_j) \ge \hat v_i(\bfx_k, j) \ge v_i^-(I^*_j)$.
\end{lemma}
\begin{proof}
Note that $y^{j-1} \le y^j$ for each $j=2,\ldots,n-1$; otherwise, $y^{j-1}=x^{j-1}_{k'} > x^j_{k''}=y^j$ for some $k',k'' \in [n]$, which means that $x^{j-1}_{k'} \ge x^j_{k''}+ 1 > x^j_{k'}$ (the last inequality holds since each knife moves only in half-step, i.e., $|x^j_{k'} - x^j_{k''}| \le \frac12$), a contradiction. Now, consider each bundle separately. For each case except for the last one, we provide a figure to illustrate the possible knife positions. 
	\begin{enumerate}
		\item Suppose $j = 1$. We have the following cases. 
		\begin{enumerate}
			\item $y^1 \not \in I_1(\bfx_k)$ so that $y^1$ does not fully appear in the first bundle of division $\calI(\bfx_k)$. Then $\hat v_i(\bfx_k, 1)=v_i(I_1(\bfx_k))=v_i(B_1)$. Thus, $v_i(I^*_1) \ge \hat v_i(\bfx_k, 1) = v_i(B_1) \ge v_i^-(I^*_1)$, since $I^*_1$ is either $B_1$ or $B_1 \cup \{y^1\}$.
			\item $y^1 \in I_1(\bfx_k)$ so that $y^1$ fully appears in the first bundle of division $\calI(\bfx_k)$. Then $I^*_1 = B_1 \cup \{ y^1 \}$ and $\hat v_i(\bfx_k, 1)=v_i(I_1(\bfx_k))=v_i(B_1 \cup \{y^1\})$. Thus, $v_i(I^*_1) = \hat v_i(\bfx_k, 1) \ge v_i^-(I^*_1)$.
		\end{enumerate}
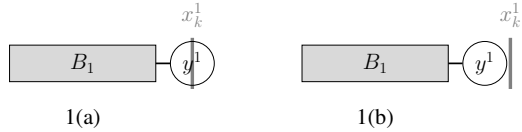
\def\stromquistscale{0.81}
\begin{center}
	\scalebox{\stromquistscale}{
		\begin{tikzpicture}[scale=0.6, transform shape, every node/.style={minimum size=12mm, inner sep=1pt}]

        \draw[ultra thick, gray] (5,1.2) node[above] {\Large $x^1_k$} -- (5,-0.2);
        \draw[fill=gray!30] (0,0) rectangle (4,1);
        \node[draw,circle](y1) at (5,0.5) {\Large $y^1$};
        \draw[-, >=latex,thick] (4,0.5)--(y1); 
        \node (label) at (2,-1) {\Large 1(a)};
	    \node at (2,0.5) {\Large $B_1$};

		\begin{scope}[shift={(8,0)}]
        \draw[ultra thick, gray] (5.7,1.2) node[above] {\Large $x^1_k$} -- (5.7,-0.2);
		\draw[fill=gray!30] (0,0) rectangle (4,1);
        \node[draw,circle](y1) at (5,0.5) {\Large $y^1$};
        \draw[-, >=latex,thick] (4,0.5)--(y1); 
		\node at (2,0.5) {\Large $B_1$};
        \node (label) at (2,-1) {\Large 1(b)};
		
		
		\end{scope}
		
		\end{tikzpicture}}
		\vspace{-3pt}
\end{center}

		\item Suppose $j = n$. Recall that $x^{n}_k= m + \frac{1}{2}$ so that $r_n(\bfx_k)= \ceil{x^{n}_k-\frac{1}{2}} =m$. If $I_n(\bfx_k) = \emptyset$, then $v_i(\bfx_k, n)=0$; thus, $v_i(I^*_n) \ge \hat v_i(\bfx_k, n) \ge v_i^-(I^*_n)= v_i(\emptyset)=0$, since $I^*_n$ is either $\emptyset$ or $\{m\}$. Suppose that $I_n(\bfx_k) \neq \emptyset$. Then $r_n(\bfx_k) \in I_n(\bfx_k)$ and hence $\hat v_i(\bfx_k, n)= v_i(I_n(\bfx_k) \setminus \{\ell_n(\bfx_k) \})$ or $\hat v_i(\bfx_k, n)= v_i(I_n(\bfx_k))$. Consider the following cases. 
		\begin{enumerate}
			\item $I_n(\bfx_k) \neq \emptyset$ and $x_k^{n-1} \le y^{n-1}$ so that $\hat v_i(\bfx_k, n)= v_i(B_n)$. Thus, $v_i(I^*_n) \ge \hat v_i(\bfx_k, n) = v_i(B_n) \ge v_i^-(I^*_n)$, since $I^*_n$ is either $B_n$ or $\{y^{n-1}\} \cup B_n  $.
			\item $I_n(\bfx_k) \neq \emptyset$ and $x_k^{n-1} =y^{n-1} + \frac12$ so that $\hat v_i(\bfx_k, n)= v_i(B_n \setminus \{y^{n-1}+1\})$. Thus, $y^{n-1}$ fully appears in the $n-1$-th bundle $I_{n-1}(\bfx_{k'})$ for some $\bfx_{k'}$ and $I^*_n = B_n$; so $v_i(I^*_n) \ge \hat v_i(\bfx_k, n) \ge v_i^-(I^*_n)$.
		\end{enumerate}
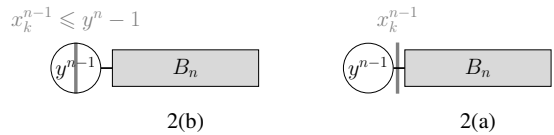
\begin{center}
	\scalebox{\stromquistscale}{
		\begin{tikzpicture}[scale=0.6, transform shape, every node/.style={minimum size=12mm, inner sep=1pt}]
		
		\begin{scope}[shift={(8,0)}]
        \draw[ultra thick, gray] (-0.2,1.2) node[above] {\Large $x^{n-1}_k$} -- (-0.2,-0.2);
        \draw[fill=gray!30] (0,0) rectangle (4,1);
        \node[draw,circle](y1) at (-1,0.5) {\Large $y^{n-1}$};
        \draw[-, >=latex,thick] (0,0.5)--(y1); 
        \node (label) at (2,-1) {\Large 2(a)};
	    \node at (2,0.5) {\Large $B_n$};
		\end{scope}

        \begin{scope}[shift={(0,0)}]
        \draw[ultra thick, gray] (-1,1.2) node[above] {\Large $x^{n-1}_k \le y^n-1$} -- (-1,-0.2);
        
		\draw[fill=gray!30] (0,0) rectangle (4,1);
        \node[draw,circle](y1) at (-1,0.5) {\Large $y^{n-1}$};
        \draw[-, >=latex,thick] (0,0.5)--(y1); 
	    \node at (2,0.5) {\Large $B_n$};
        \node (label) at (2,-1) {\Large 2(b)};		
		\end{scope}
		
		\end{tikzpicture}}
		\vspace{-3pt}
\end{center}

		\item Suppose $j \in \{2,3,\ldots,n-1\}$, $y^{j-1} < y^j$, and $I^*_j =B_j$. By the {\bf if}-condition in Line \ref{line:if}, this means that $y^{j-1}$ has been allocated to the $j-1$st bundle $I^*_{j-1}$ and there is no partial division represented by the vertices of $S$ such that $y^j$ fully appears in the $j$-th bundle. Hence, by $(b)$ of Line \ref{line:if}, $x^{j-1}_{k'}=y^{j-1}+\frac 12$ and $x^{j}_{k'}= y^{j}-\frac 12$ 
		for some vertex $\bfx_{k'}$, meaning that $y^{j-1} \le x_k^{j-1}$ and $x_k^{j} \le  y^{j}$ since each knife moves in half-step only. Further, the case when $x_k^{j-1} = y^{j-1}$ and $x_k^{j} = y^{j}$ is not possible: indeed, if so, we have $\phi^{-1}(k)<\phi^{-1}(k')$ by the fact that $x_k^{j-1} = y^{j-1} <y^{j-1} +\frac 12= x^{j-1}_{k'}$ and by \eqref{eq:sperner:kuhn-triangulation}, which again implies $y^j=x_k^{j} \le x^j_{k'}= y^{j} -\frac12$ by \eqref{eq:sperner:kuhn-triangulation}, a contradiction. We thus have either
		\begin{enumerate}
  			\item $x_k^{j-1} = y^{j-1}$ and $x_k^{j} = y^{j}-\frac12$ so that $\hat v_i(\bfx_k, j)=v_i(I_j(\bfx_k))=v_i(B_j)$.
			\item $x_k^{j-1} = y^{j-1} + \frac12$ and $x_k^{j} = y^{j}-\frac12$ so that $\hat v_i(\bfx_k, j)=v_i(I_j(\bfx_k)\setminus \{\ell_{j}(\bfx_k)
			\})=v_i(B_j\setminus \{y^{j-1}+1\})$.
			\item $x_k^{j-1} = y^{j-1} + \frac12$ and $x_k^{j} = y^{j}$ so that $\hat v_i(\bfx_k, j)=v^-_i(B_j\cup \{y^j\})$.

		\end{enumerate}
		In either case, $v_i(I^*_j)=v_i(B_j) \ge \hat v_i(\bfx_k, j) \ge v_i^-(B_j)=v_i^-(I^*_j)$.
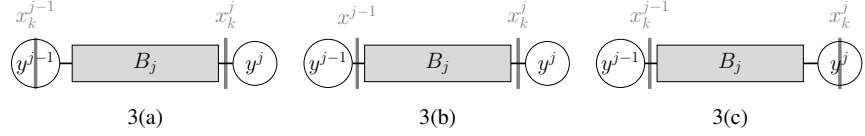
\begin{center}
	\scalebox{\stromquistscale}{
		\begin{tikzpicture}[scale=0.6, transform shape, every node/.style={minimum size=12mm, inner sep=1pt}]

        \begin{scope}[shift={(8,0)}]
        \draw[ultra thick, gray] (-0.2,1.2) node[above] {\Large $x^{j-1}$} -- (-0.2,-0.2);
        \draw[ultra thick, gray] (4.2,1.2) node[above] {\Large $x^j_k$} -- (4.2,-0.2);
        
        \draw[fill=gray!30] (0,0) rectangle (4,1);
        \node[draw,circle](y1) at (-1,0.5) {\Large $y^{j-1}$};
        \node[draw,circle](y2) at (5,0.5) {\Large $y^j$};
        \draw[-, >=latex,thick] (0,0.5)--(y1); 
        \draw[-, >=latex,thick] (4,0.5)--(y2); 
        \node (label) at (2,-1) {\Large 3(b)};
	    \node at (2,0.5) {\Large $B_j$};
        \end{scope}
		
		\begin{scope}[shift={(16,0)}]
        \draw[ultra thick, gray] (-0.2,1.2) node[above] {\Large $x^{j-1}_k$} -- (-0.2,-0.2);
        \draw[ultra thick, gray] (5,1.2) node[above] {\Large $x^j_k$} -- (5,-0.2);
        \draw[fill=gray!30] (0,0) rectangle (4,1);
        \node[draw,circle](y1) at (-1,0.5) {\Large $y^{j-1}$};
        \node[draw,circle](y2) at (5,0.5) {\Large $y^j$};
        \draw[-, >=latex,thick] (0,0.5)--(y1); 
        \draw[-, >=latex,thick] (4,0.5)--(y2); 
        \node (label) at (2,-1) {\Large 3(c)};
	    \node at (2,0.5) {\Large $B_j$};
		\end{scope}

        \begin{scope}[shift={(0,0)}]
        \draw[ultra thick, gray] (-1,1.2) node[above] {\Large $x^{j-1}_k$} -- (-1,-0.2);
        \draw[ultra thick, gray] (4.2,1.2) node[above] {\Large $x^j_k$} -- (4.2,-0.2);
        
		\draw[fill=gray!30] (0,0) rectangle (4,1);
        \node[draw,circle](y1) at (-1,0.5) {\Large $y^{j-1}$};
        \node[draw,circle](y2) at (5,0.5) {\Large $y^j$};
        \draw[-, >=latex,thick] (0,0.5)--(y1); 
        \draw[-, >=latex,thick] (4,0.5)--(y2); 
	    \node at (2,0.5) {\Large $B_j$};
        \node (label) at (2,-1) {\Large 3(a)};		
		\end{scope}
		
		\end{tikzpicture}}
		\vspace{-3pt}
\end{center}		
		\item Suppose $j \in \{2,3,\ldots,n-1\}$, $y^{j-1} < y^j$, and $I^*_j =\{ y^{j-1} \}  \cup B_j$. Since $y^{j-1}$ has not been allocated to $I^*_{j-1}$, there is no partial division represented by the vertices of $S$ such that $y^{j-1}$ fully appears in the $j-1$st bundle by the {\bf if}-condition in Line \ref{line:if}, which means that $y^{j-1}$ fully appears in the $j$-th bundle $I_j(\bfx_{k'})$ for some ${k'} \in [n]$. On the other hand, since $y^j$ has not been allocated to $I^*_{j}$, there is no partial division represented by the vertices of $S$ such that $y^{j}$ fully appears in the $j$-th bundle. Thus, we have $x_k^{j-1} \le y^{j-1}$ and $x_k^j \le y^j$. Consider the following cases: 
		\begin{enumerate}
            \item $x_k^{j-1} = y^{j-1}-\frac12$ and $x_k^j =y^j -\frac12$ so that $\hat v_i(\bfx_k,j)= v_i(I_j(\bfx_k)\setminus \{\ell_{j}(\bfx_k)\})=v_i(B_j)$. 
            \item $x_k^{j-1} = y^{j-1}-\frac12$ and $x_k^j =y^j $ so that $\hat v_i(\bfx_k,j)=v_i^-(I_j(\bfx_k)\cup\{r_{j}(\bfx_k)\})=v^-_i(\{y^{j-1}\} \cup B_j \cup \{y^j\})$. 
            \item $x_k^{j-1} = y^{j-1}$ and $x_k^j =y^j -\frac12$ so that $\hat v_i(\bfx_k,j)=v_i(I_j(\bfx_k))=v_i(B_j)$. 
			\item $x_k^{j-1} = y^{j-1}$ and $x_k^j =y^j$ so that $\hat v_i(\bfx_k,j)=v^-_i(\{y^{j-1}\} \cup B_j \cup \{y^j\})$.
		\end{enumerate}
		In either case, $v_i(I^*_j)=v_i(\{ y^{j-1} \}  \cup  B_j) \ge \hat v_i(\bfx_k, j) \ge v_i^-(\{y^{j-1}\} \cup  B_j )=v_i^-(I^*_j)$.

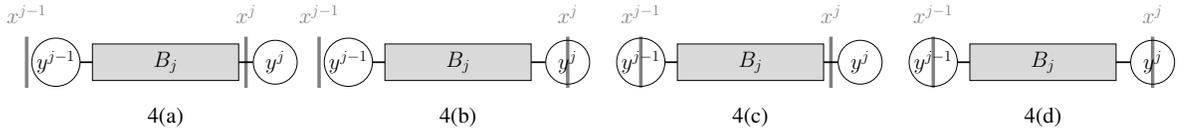
\begin{center}
	\scalebox{\stromquistscale}{
		\begin{tikzpicture}[scale=0.6, transform shape, every node/.style={minimum size=12mm, inner sep=1pt}]

		\begin{scope}[shift={(-2,0)}]
        \draw[ultra thick, gray] (-1.8,1.2) node[above] {\Large $x^{j-1}$} -- (-1.8,-0.2);
        \draw[ultra thick, gray] (4.2,1.2) node[above] {\Large $x^j$} -- (4.2,-0.2);
        \draw[fill=gray!30] (0,0) rectangle (4,1);
        \node[draw,circle](y1) at (-1,0.5) {\Large $y^{j-1}$};
        \node[draw,circle](y2) at (5,0.5) {\Large $y^j$};
        \draw[-, >=latex,thick] (0,0.5)--(y1); 
        \draw[-, >=latex,thick] (4,0.5)--(y2); 
        \node (label) at (2,-1) {\Large 4(a)};
	    \node at (2,0.5) {\Large $B_j$};
		\end{scope}

        \begin{scope}[shift={(6,0)}]
        \draw[ultra thick, gray] (-1.8,1.2) node[above] {\Large $x^{j-1}$} -- (-1.8,-0.2);
        \draw[ultra thick, gray] (5,1.2) node[above] {\Large $x^j$} -- (5,-0.2);
        
		\draw[fill=gray!30] (0,0) rectangle (4,1);
        \node[draw,circle](y1) at (-1,0.5) {\Large $y^{j-1}$};
        \node[draw,circle](y2) at (5,0.5) {\Large $y^j$};
        \draw[-, >=latex,thick] (0,0.5)--(y1); 
        \draw[-, >=latex,thick] (4,0.5)--(y2); 
	    \node at (2,0.5) {\Large $B_j$};
        \node (label) at (2,-1) {\Large 4(b)};		
		\end{scope}

        \begin{scope}[shift={(14,0)}]
        \draw[ultra thick, gray] (-1,1.2) node[above] {\Large $x^{j-1}$} -- (-1,-0.2);
        \draw[ultra thick, gray] (4.2,1.2) node[above] {\Large $x^j$} -- (4.2,-0.2);
        
		\draw[fill=gray!30] (0,0) rectangle (4,1);
        \node[draw,circle](y1) at (-1,0.5) {\Large $y^{j-1}$};
        \node[draw,circle](y2) at (5,0.5) {\Large $y^j$};
        \draw[-, >=latex,thick] (0,0.5)--(y1); 
        \draw[-, >=latex,thick] (4,0.5)--(y2); 
	    \node at (2,0.5) {\Large $B_j$};
        \node (label) at (2,-1) {\Large 4(c)};		
		\end{scope}

          \begin{scope}[shift={(22,0)}]
        \draw[ultra thick, gray] (-1,1.2) node[above] {\Large $x^{j-1}$} -- (-1,-0.2);
        \draw[ultra thick, gray] (5,1.2) node[above] {\Large $x^j$} -- (5,-0.2);
        
        \draw[fill=gray!30] (0,0) rectangle (4,1);
        \node[draw,circle](y1) at (-1,0.5) {\Large $y^{j-1}$};
        \node[draw,circle](y2) at (5,0.5) {\Large $y^j$};
        \draw[-, >=latex,thick] (0,0.5)--(y1); 
        \draw[-, >=latex,thick] (4,0.5)--(y2); 
        \node (label) at (2,-1) {\Large 4(d)};
	    \node at (2,0.5) {\Large $B_j$};
        \end{scope}
		\end{tikzpicture}}
		\vspace{-3pt}
\end{center}

		\item Suppose $j \in \{2,3,\ldots,n-1\}$, $y^{j-1}< y^j$, and $I^*_j = B_j \cup \{ y^{j} \}$. Since $y^j$ is allocated to the $j$-th bundle but $y^{j-1}$ is allocated to some other bundle under $\calI^*$, there is no $k'$ with $x^{j-1}_{k'}=y^{j-1}+\frac12$ and $x^j_{k'} =y^{j}-\frac12$ by the {\bf if}-condition $($b$)$ in Line \ref{line:if}. Then we have the following cases:
		\begin{enumerate}
			\item $x_k^{j-1} \le y^{j-1}$ and $x_k^j =y^j -\frac12$ so that $\hat v_i(\bfx_k,j)=v_i(B_j)$. 
			\item $x_k^{j-1} \le y^{j-1}$ and $x_k^j =y^j$ so that $\hat v_i(\bfx_k,j)=v^-_i(\{y^{j-1}\} \cup B_j \cup \{y^j\})$.
			\item $x_k^{j-1} \le y^{j-1}$ and $x_k^j =y^j+ \frac12$ so that $\hat v_i(\bfx_k,j)=v_i(B_j \cup \{y^j\})$.
			\item $x_k^{j-1} = y^{j-1} + \frac12$ and $x_k^j =y^j$ so that $\hat v_i(\bfx_k,j)=v^-_i(B_j \cup \{y^j\})$.
			\item $x_k^{j-1} = y^{j-1} + \frac12$ and $x_k^j =y^j + \frac12$ so that $\hat v_i(\bfx_k,j)=v_i(I_j(\bfx_k)\setminus \{\ell_{j}(\bfx_k)\})=v_i((B_j \setminus \{y^{j-1}+1\}) \cup \{y^j\})$
		\end{enumerate}
		In either case, $v_i(I^*_j)=v_i(B_j \cup \{y^j\}) \ge \hat v_i(\bfx_k, j) \ge v_i^-(B_j \cup \{y^j\})=v_i^-(I^*_j)$.

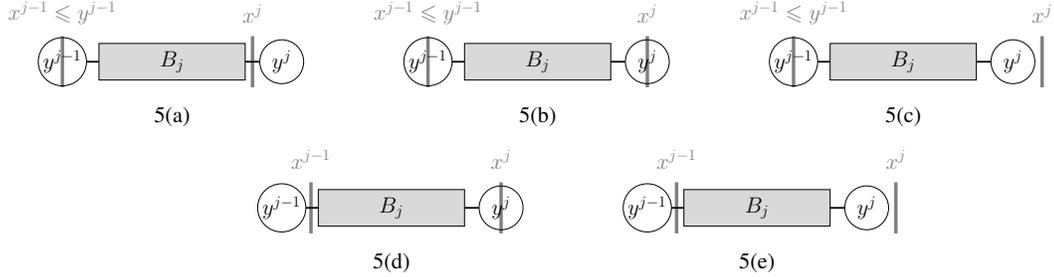
\begin{center}
	\scalebox{\stromquistscale}{
		\begin{tikzpicture}[scale=0.6, transform shape, every node/.style={minimum size=12mm, inner sep=1pt}]

         \draw[ultra thick, gray] (-1,1.2) node[above] {\Large $x^{j-1} \le y^{j-1}$} -- (-1,-0.2);
        \draw[ultra thick, gray] (4.2,1.2) node[above] {\Large $x^j$} -- (4.2,-0.2);
        
        \draw[fill=gray!30] (0,0) rectangle (4,1);
        \node[draw,circle](y1) at (-1,0.5) {\Large $y^{j-1}$};
        \node[draw,circle](y2) at (5,0.5) {\Large $y^j$};
        \draw[-, >=latex,thick] (0,0.5)--(y1); 
        \draw[-, >=latex,thick] (4,0.5)--(y2); 
        \node (label) at (2,-1) {\Large 5(a)};
	    \node at (2,0.5) {\Large $B_j$};

		\begin{scope}[shift={(10,0)}]
        \draw[ultra thick, gray] (-1,1.2) node[above] {\Large $x^{j-1} \le y^{j-1}$} -- (-1,-0.2);
        \draw[ultra thick, gray] (5,1.2) node[above] {\Large $x^j$} -- (5,-0.2);
        \draw[fill=gray!30] (0,0) rectangle (4,1);
        \node[draw,circle](y1) at (-1,0.5) {\Large $y^{j-1}$};
        \node[draw,circle](y2) at (5,0.5) {\Large $y^j$};
        \draw[-, >=latex,thick] (0,0.5)--(y1); 
        \draw[-, >=latex,thick] (4,0.5)--(y2); 
        \node (label) at (2,-1) {\Large 5(b)};
	    \node at (2,0.5) {\Large $B_j$};
		\end{scope}

        \begin{scope}[shift={(20,0)}]
        \draw[ultra thick, gray] (-1,1.2) node[above] {\Large $x^{j-1} \le y^{j-1}$} -- (-1,-0.2);
        \draw[ultra thick, gray] (5.8,1.2) node[above] {\Large $x^j$} -- (5.8,-0.2);
        
		\draw[fill=gray!30] (0,0) rectangle (4,1);
        \node[draw,circle](y1) at (-1,0.5) {\Large $y^{j-1}$};
        \node[draw,circle](y2) at (5,0.5) {\Large $y^j$};
        \draw[-, >=latex,thick] (0,0.5)--(y1); 
        \draw[-, >=latex,thick] (4,0.5)--(y2); 
	    \node at (2,0.5) {\Large $B_j$};
        \node (label) at (2,-1) {\Large 5(c)};		
		\end{scope}

        \begin{scope}[shift={(6,-4)}]
        \draw[ultra thick, gray] (-0.2,1.2) node[above] {\Large $x^{j-1}$} -- (-0.2,-0.2);
        \draw[ultra thick, gray] (5,1.2) node[above] {\Large $x^j$} -- (5,-0.2);
        
		\draw[fill=gray!30] (0,0) rectangle (4,1);
        \node[draw,circle](y1) at (-1,0.5) {\Large $y^{j-1}$};
        \node[draw,circle](y2) at (5,0.5) {\Large $y^j$};
        \draw[-, >=latex,thick] (0,0.5)--(y1); 
        \draw[-, >=latex,thick] (4,0.5)--(y2); 
	    \node at (2,0.5) {\Large $B_j$};
        \node (label) at (2,-1) {\Large 5(d)};		
		\end{scope}

        \begin{scope}[shift={(16,-4)}]
        \draw[ultra thick, gray] (-0.2,1.2) node[above] {\Large $x^{j-1}$} -- (-0.2,-0.2);
        \draw[ultra thick, gray] (5.8,1.2) node[above] {\Large $x^j$} -- (5.8,-0.2);
        
		\draw[fill=gray!30] (0,0) rectangle (4,1);
        \node[draw,circle](y1) at (-1,0.5) {\Large $y^{j-1}$};
        \node[draw,circle](y2) at (5,0.5) {\Large $y^j$};
        \draw[-, >=latex,thick] (0,0.5)--(y1); 
        \draw[-, >=latex,thick] (4,0.5)--(y2); 
	    \node at (2,0.5) {\Large $B_j$};
        \node (label) at (2,-1) {\Large 5(e)};		
		\end{scope}
		\end{tikzpicture}}
		\vspace{-3pt}
\end{center}
		\item Suppose $j \in \{2,3,\ldots,n-1\}$, $y^{j-1} < y^j$, and $I^*_j = \{ y^{j-1} \} \cup B_j \cup \{ y^j \}$. Then, there exist some ${k'},k'' \in [n]$ such that $y^{j-1}$ and $y^j$ fully appear in the $j$-th bundles $I_j(\bfx_{k'})$ and $I_j(\bfx_{k''})$, respectively. We thus have $x_k^{j-1} \le y^{j-1}$ and $y^j \le x_k^j$. Consider the following cases: 
		\begin{enumerate}

  			 \item $x_k^{j-1} = y^{j-1} - \frac12$ and $x_k^j =y^j$ so that $\hat v_i(\bfx_k, j) =v_i^-(I_j(\bfx_k)\cup\{r_{j}(\bfx_k)\})  =  v_i^-(\{y^{j-1}\} \cup B_j \cup \{y^j\})$. 
            \item $x_k^{j-1} = y^{j-1} - \frac12$ and $x_k^j=y^j+\frac12$ so that $\hat v_i(\bfx_k, j) =v_i(I_j(\bfx_k) \setminus \{\ell_{j}(\bfx_k) \})= v_i(B_j \cup \{y^j\})$. 
            \item $x_k^{j-1} = y^{j-1}$ and $x_k^j =y^j$ so that $\hat v_i(\bfx_k, j) = v_i^-(I_j(\bfx_k)\cup\{\ell_{j}(\bfx_k),r_{j}(\bfx_k)\})= v_i^-(\{y^{j-1}\} \cup B_j \cup \{y^j\})$. 
            \item $x_k^{j-1} = y^{j-1} $ and $x_k^j=y^j+\frac12$ so that $\hat v_i(\bfx_k, j) = v_i(I_j(\bfx_k))= v_i(B_j \cup \{y^j\})$.

		\end{enumerate}
		In either case, $v_i(I^*_j) \ge \hat v_i(\bfx_k, j) \ge v_i^-(I^*_j)$ since $I^*_j = \{y^{j-1}\} \cup B_j \cup \{y^j\}$.

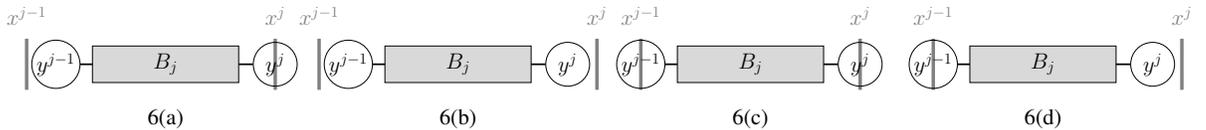
\begin{center}
	\scalebox{\stromquistscale}{
		\begin{tikzpicture}[scale=0.6, transform shape, every node/.style={minimum size=12mm, inner sep=1pt}]
        \begin{scope}[shift={(-2,0)}]
        \draw[ultra thick, gray] (-1.8,1.2) node[above] {\Large $x^{j-1}$} -- (-1.8,-0.2);
        \draw[ultra thick, gray] (5,1.2) node[above] {\Large $x^j$} -- (5,-0.2);
        
		\draw[fill=gray!30] (0,0) rectangle (4,1);
        \node[draw,circle](y1) at (-1,0.5) {\Large $y^{j-1}$};
        \node[draw,circle](y2) at (5,0.5) {\Large $y^j$};
        \draw[-, >=latex,thick] (0,0.5)--(y1); 
        \draw[-, >=latex,thick] (4,0.5)--(y2); 
	    \node at (2,0.5) {\Large $B_j$};
        \node (label) at (2,-1) {\Large 6(a)};		
		\end{scope}

 \begin{scope}[shift={(14,0)}]
        \draw[ultra thick, gray] (-1,1.2) node[above] {\Large $x^{j-1}$} -- (-1,-0.2);
        \draw[ultra thick, gray] (5,1.2) node[above] {\Large $x^j$} -- (5,-0.2);
        
        \draw[fill=gray!30] (0,0) rectangle (4,1);
        \node[draw,circle](y1) at (-1,0.5) {\Large $y^{j-1}$};
        \node[draw,circle](y2) at (5,0.5) {\Large $y^j$};
        \draw[-, >=latex,thick] (0,0.5)--(y1); 
        \draw[-, >=latex,thick] (4,0.5)--(y2); 
        \node (label) at (2,-1) {\Large 6(c)};
	    \node at (2,0.5) {\Large $B_j$};
        \end{scope}

      	\begin{scope}[shift={(6,0)}]
        \draw[ultra thick, gray] (-1.8,1.2) node[above] {\Large $x^{j-1}$} -- (-1.8,-0.2);
        \draw[ultra thick, gray] (5.8,1.2) node[above] {\Large $x^j$} -- (5.8,-0.2);
        \draw[fill=gray!30] (0,0) rectangle (4,1);
        \node[draw,circle](y1) at (-1,0.5) {\Large $y^{j-1}$};
        \node[draw,circle](y2) at (5,0.5) {\Large $y^j$};
        \draw[-, >=latex,thick] (0,0.5)--(y1); 
        \draw[-, >=latex,thick] (4,0.5)--(y2); 
        \node (label) at (2,-1) {\Large 6(b)};
	    \node at (2,0.5) {\Large $B_j$};
		\end{scope}

        \begin{scope}[shift={(22,0)}]
        \draw[ultra thick, gray] (-1,1.2) node[above] {\Large $x^{j-1}$} -- (-1,-0.2);
        \draw[ultra thick, gray] (5.8,1.2) node[above] {\Large $x^j$} -- (5.8,-0.2);
        
		\draw[fill=gray!30] (0,0) rectangle (4,1);
        \node[draw,circle](y1) at (-1,0.5) {\Large $y^{j-1}$};
        \node[draw,circle](y2) at (5,0.5) {\Large $y^j$};
        \draw[-, >=latex,thick] (0,0.5)--(y1); 
        \draw[-, >=latex,thick] (4,0.5)--(y2); 
	    \node at (2,0.5) {\Large $B_j$};
        \node (label) at (2,-1) {\Large 6(d)};		
		\end{scope}
	
		\end{tikzpicture}}
		\vspace{-3pt}
\end{center}
		\item Suppose that $j \in \{2,3,\ldots,n-1\}$ and $y^{j-1}=y^j$. This means that $y^{j-1}=x^{j-1}_{k'}$ and $y^j=x^j_{k''}$ for some $k', k''$. Since each knife moves in half-step, $y^j=x^{j-1}_{k'} \le x^j_{k'} \le x^j_{k''} + \frac 12= y^j + \frac 12$. Thus, we have  $x^j_{k'} -x^{j-1}_{k'} \le \frac 12$ and $B_j = \emptyset$. Consider the following cases. 
		\begin{enumerate}
            \item $|x^{j}_k - x^{j-1}_{k}| \le \frac 12$ so that $I_j(\bfx_k)=\emptyset$ and $\hat v_i(\bfx_k, j)=0$. 
            \item  $x^{j-1}_{k}= y^j - \frac{1}{2} $ and $x^{j}_k= y^j + \frac{1}{2}$ so that $I_j(\bfx_k)=\{y^{j}\}$ and $\hat v_i(\bfx_k, j)= v_i(I_j(\bfx_k) \setminus \{\ell_{j}(\bfx_k) \})= v_i (\{y^{j}\}\setminus \{y^{j}\})=v_i(\emptyset)=0$.
		\end{enumerate}
In either case, $v_i(I^*_j) \ge \hat v_i(\bfx_k, j) \ge v_i^-(I^*_j)$, since $v^{-}_i(I^*_j)=v_i(\emptyset)=0$. 

\end{enumerate}
This completes the proof.~\qed
\end{proof}

We are now ready to prove Theorem \ref{thm:EF1:general}. 

\begin{proof}[of Theorem \ref{thm:EF1:general}]
Applying Theorem \ref{thm:sperner} to the triangulation $\ttT_{\half}$ with the coloring function $\lambda$, we obtain an elementary simplex $S^*=\langle \bfx^*_1,\bfx^*_2,\ldots,\bfx^*_n \rangle$ of $\ttT_{\half}$. On this simplex, there exists a permutation $\pi \colon [n] \rightarrow [n]$ such that for each $i \in [n]$, $\pi(i) \in \lambda(\bfx^*_i)$.
That is, for each $i\in [n]$, the bundle with index $\pi(i)$ in the division $\calI(\bfx^*_i)$ is the bundle most preferred by the owner $a(\bfx^*_i)$:
\begin{equation}
	\label{eq:ef2:my-bundle-is-best}
	{\hat v}_{a(\bfx^*_i)} (\bfx^*_i,\pi(i)) \ge {\hat v}_{a(\bfx^*_i)} (\bfx^*_i,j) \quad\text{for each $j\in[n]$}.
\end{equation}

Without loss of generality, we assume that $a(\bfx^*_i)=i$ for each $i \in [n]$. Applying Algorithm \ref{alg:rounding} to $S^*$, we obtain a division $\calI^*=(I^*_1,I^*_2,...,I^*_n)$. For every pair of agents $i,j \in [n]$, we have
\begin{alignat*}{2}
v_i(I^*_{\pi(i)})
&\ge \hat v_i (\bfx^*_i, \pi(i)) &&\text{by Lemma~\ref{lem:ef1:approx-correct}},
\\
&\ge \hat v_i (\bfx^*_i, j) &&\text{since $\pi(i) \in \smash{\lambda_{i}}(\bfx^*_{i})$},
\\
&\ge v_i^-(\smash{I^*_{j}}) \qquad \quad &&\text{by Lemma~\ref{lem:ef1:approx-correct}}.
\end{alignat*}
Thus, $\pi$ certifies that $\calI^*$ is a desired division.~\qed
\end{proof}

\section{Secretive and extra versions}\label{sec:sec:extra}
In this section, we prove the secretive and extra versions of EF1 existence. 
A {\em secretive EF1$_{outer}$ division} for agent $i^*$ is a division $\calI=(I_1,I_2,\ldots,I_n)$ of a path into $n$ connected subsets where whichever part a secretive agent $i^*$ selects, an EF1$_{\emph outer}$ assignment of the remaining bundles can be made to the other agents, i.e., for every index $j \in [n]$, there exists a bijection $\pi\colon[n]\setminus \{i^*\} \rightarrow [n] \setminus \{j\}$ such that for every non-secretive agent $i \in [n] \setminus \{i^*\}$, 
$$
v_i(\calI_{\pi(i)}) \ge \max_{j' \in [n]}v^{-}_{i} (\calI_{j'}).
$$

For $n+1$ agents with monotone valuations, an {\em extra} EF1$_{\emph outer}$ division is a division $\calI=(I_1,I_2,\ldots,I_{n})$ of a path into $n$ connected subsets when any extra agent $i^*$ leaves, an EF1$_{\emph outer}$ assignment of the bundles can be made to the remaining agents, i.e., for every extra agent $i^* \in [n+1]$, there exists a bijection $\pi\colon[n+1]\setminus \{i^*\} \rightarrow [n]$ such that for every remaining agent $i \in [n+1] \setminus \{i^*\}$,
$$
v_i(\calI_{\pi(i)}) \ge \max_{j \in [n]}v^{-}_i (\calI_{j}).
$$ 

The main theorems of this section are as follows: 

\begin{theorem}\label{thm:EF1:secretive}
Suppose that there are $n$ agents with monotone valuations over connected bundles of a path. Then, for any agent $i^* \in [n]$, there exists a secretive EF1$_\emph{outer}$ division for $i^*$.
\end{theorem}

\begin{theorem}\label{thm:EF1:extra}
Suppose that there are $n+1$ agents with monotone valuations over connected bundles of a path. Then, there exists an extra EF1$_\emph{ outer}$ division.
\end{theorem}

To establish the above results, we use a more general version of Sperner's lemma for multiple proper colorings, proved by \citet{MeunierSu}. Consider a triangulation $\ttT$ of the $(n-1)$ standard simplex and coloring functions $\lambda_i \colon V(\ttT) \rightarrow 2^{[n]}$ for $i \in [n+1]$. For each elementary simplex $S \in \ttT$, we define the associated bipartite graph representation $G(S)$. The left and right vertices of graph $G(S)$ correspond to the coloring functions $\lambda_1,\lambda_2,\ldots,\lambda_{n+1}$ and colors $[n]$, respectively. There is an edge $\{ \lambda_i,j \}$ if and only if $j \in \lambda_i(\bfx)$ for some main vertex $\bfx$ of the elementary simplex $S$. An example of such bipartite graph representation is given in Figure~\ref{fig:bipartite}. The following multi-labeled version of Sperner's lemma was shown in their proof of Theorem $2.2$. 

\begin{theorem}[Corollary of the proof of Theorem $2.2$ in \citet{MeunierSu}]\label{thm:multisperner}
Let $\ttT$ be a triangulation of the $(n-1)$-standard simplex and let $\lambda_1,\lambda_2,\dots,\lambda_{n+1}$ be proper colorings on $\ttT$. Then the following hold:
\begin{enumerate}
\item[(1)] There exists an elementary simplex $S^*=\langle {\bfx}_1,{\bfx}_2,\ldots,{\bfx}_n \rangle \in \ttT$ for which for any vertex $j$ of $[n]$, the graph $G(S^*)$ has a perfect matching between the left vertices in $\{\lambda_1, \lambda_2,\ldots,\lambda_{n-1}\}$ and the right vertices in $[n] \setminus \{j\}$.  
\item[(2)] There exists an elementary simplex $S^*=\langle {\bfx}_1,{\bfx}_2,\ldots,{\bfx}_{n} \rangle \in \ttT$ for which for any $\lambda_{i^*}$, the graph $G(S^*)$ has a perfect matching between the right vertices in $[n]$ and the left vertices in $\{\lambda_1, \lambda_2,\ldots,\lambda_{n+1}\}\setminus \{\lambda_{i^*}\}$. 
\end{enumerate}
\end{theorem}

\begin{figure*}[htb]
\begin{subfigure}[t]{0.5\columnwidth}
\centering
\begin{tikzpicture}[scale=0.7, transform shape]
\draw [thick] (0,0) node[below]{${\{1\}},\{1\},{\{2\}},\{1,2\}$} -- (4,0);
\draw [thick] (0,0) node[below]{} -- (2,3) node[above]{${\{1\}},\{1\},{\{2,3\}},\{2\}$};
\draw [thick] (4,0) node[below]{${\{1\}},\{1,2\},{\{2\}},\{2\}$} -- (2,3);
\end{tikzpicture}
\subcaption{The elementary simplex with coloring functions $\lambda_1,\lambda_2,\lambda_3,\lambda_4$}
\end{subfigure}
\begin{subfigure}[t]{0.5\columnwidth}
\centering
\begin{tikzpicture}[scale=0.6,thick,
  every node/.style={draw,circle},
  fsnode/.style={fill=white},
  ssnode/.style={fill=white},
  every fit/.style={ellipse,draw,inner sep=-1pt,text width=1.7cm},
  ->,shorten >= 3pt,shorten <= 3pt
]

\begin{scope}[start chain=going below,node distance=5mm]
\foreach \i in {1,2,3,4}
  \node[fsnode,on chain] (f\i) [label=left: $\lambda_{\i}$] {};
\end{scope}

\begin{scope}[xshift=4cm,yshift=-1cm,start chain=going below,node distance=5mm]
\foreach \i in {1,2,3}
  \node[ssnode,on chain] (s\i) [label=right: \i] {};
\end{scope}


\draw[-, >=latex,thick] (f1) -- (s1);
\draw[-, >=latex,thick] (f1) -- (s3);
\draw[-, >=latex,thick] (s1) -- (f2);
\draw[-, >=latex,thick] (f2) -- (s2);
\draw[-, >=latex,thick] (s2) -- (f3);
\draw[-, >=latex,thick] (s3) -- (f3);
\draw[-, >=latex,thick] (f4) -- (s2);
\draw[-, >=latex,thick] (f4) -- (s1);
\end{tikzpicture}
\subcaption{The bipartite graph representation of $($a$)$}
\end{subfigure}
\caption{Example of the bipartite graph representation of an elementary simplex with $n=3$. In Figure $($a$)$, the collection of sets around each corner vertex $\bfx_j$ represents how each $\lambda_i$ colors that vertex, i.e., $(\lambda_1(\bfx_j),\lambda_2(\bfx_j),\lambda_3(\bfx_j),\lambda_4(\bfx_j))$. }\label{fig:bipartite}
\end{figure*}
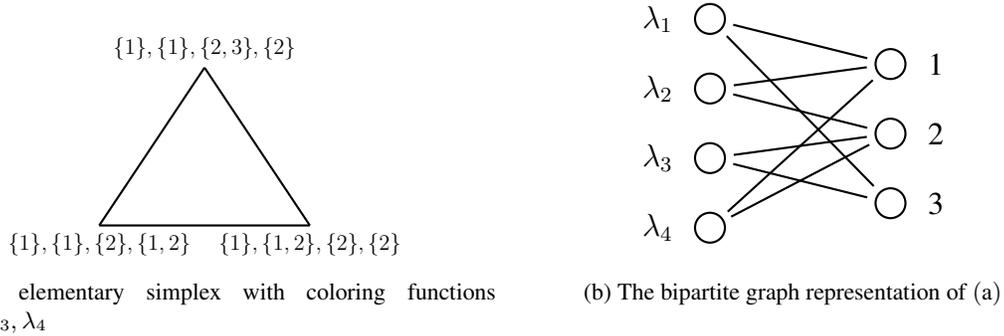

Statement $(1)$ in Theorem \ref{thm:multisperner} asserts the existence of an elementary simplex $S^*$ whose graph representation has a matching covering the vertices in $\{\lambda_1,\lambda_2,\ldots,\lambda_{n-1}\}$ after the removal of any vertex in $[n]$; for example, in Figure~\ref{fig:bipartite}, there is a perfect matching between $\{\lambda_1,\lambda_2\}$ and $\{1,2,3\}\setminus \{j\}$ for any $j \in \{1,2,3\}$. 

Statement $(2)$ in Theorem \ref{thm:multisperner} ensures the existence of an elementary simplex $S^*$ whose graph representation has a  matching covering the right vertices $[n]$ after the removal of any vertex in $\{\lambda_1, \lambda_2,\ldots,\lambda_{n+1}\}$; for example, in Figure~\ref{fig:bipartite}, there is a perfect matching between $\{1,2,3\}$ and $\{\lambda_1,\lambda_2,\lambda_3,\lambda_4\} \setminus \{\lambda_{i^*}\}$  for any $i^* \in \{1,2,3,4\}$.

Using the above theorem, \citet{MeunierSu} recovered the existence of a secretive envy-free division of a cake \citep{Woodall1980,Asada2018} and further proved its dual version that the cake can be divided into $n$ connected pieces so that no matter which agent gets kicked out, there is an envy-free assignment of the pieces to the remaining agents.   

By applying Theorem~\ref{thm:multisperner} to the $\ttT_{\half}$ in the previous section, we prove Theorems~\ref{thm:EF1:secretive} and \ref{thm:EF1:extra}.\footnote{Note that Theorem \ref{thm:multisperner} is concerned with the standard simplex $\Delta^{n-1}$. However, Theorem \ref{thm:multisperner} can be also applied to $S_m$ since $S_m$ and $\Delta^{n-1}$ are affinely equivalent (there is an affine transformation $f \colon \Delta^{n-1} \rightarrow S_m$ with the $i$-th component $f(\bfx)^{i}=mx_i + \frac{1}{2}$).} 

\begin{proof}[of Theorem~\ref{thm:EF1:secretive}]
Suppose there are $n$ agents. Assume without loss of generality that $n$ is a secretive agent, namely, $i^*=n$. We use the same simplex $S_m$, triangulation $\ttT_{\half}$, and coloring functions $\lambda_i$ $(i \in [n])$ as defined in Section~\ref{sec:EF1}. 
By applying Theorem \ref{thm:multisperner} to $\ttT_{\half}$ with $\lambda_i$ $(i \in [n])$, we obtain an elementary simplex $S^*_1=\langle \bfx^*_1,\bfx^*_2,\ldots,\bfx^*_n \rangle \in \ttT_{\half}$ satisfying Condition $($1$)$ of Theorem \ref{thm:multisperner}: that is, regardless of which $j \in [n]$ we remove from, the graph $G(S^*_1)$ has a matching covering the vertices in ${\lambda}_1,{\lambda}_2,\ldots,{\lambda}_{n-1}$. Thus, for each $j \in [n]$, there exists a bijection $\pi_j \colon [n-1] \rightarrow [n]\setminus \{j\}$ such that for every $i \in [n-1]$,
\begin{align}\label{eq:secretive}
&\pi_j(i) \in \lambda_i(\bfx^*_{k})~~\mbox{for some}~~\bfx^*_k.  
\end{align}

Now, by applying Algorithm \ref{alg:rounding} to the elementary simplex $S^*_1$, we obtain a division $\calI^*=(I^*_1,I^*_2,\ldots,I^*_n)$ of the path. Take any index $j \in [n]$ and any non-secretive agent $i \in [n-1]$. For a bijection $\pi_j \colon [n-1] \rightarrow [n]\setminus \{j\}$, there exists a main vertex $\bfx^*_k$ of $S^*_1$ where $\pi_j(i) \in \lambda_{i}(\bfx^*_{k})$. Hence, we have
\begin{alignat*}{2}
v_i(I^*_{\pi_j(i)})
&\ge \hat v_i (\bfx^*_k, \pi_j(i)) &&\text{by Lemma~\ref{lem:ef1:approx-correct}},
\\
&\ge \hat v_i (\bfx^*_k, j') &&\text{since $\pi_j(i) \in \smash{\lambda_{i}}(\bfx^*_{k})$},
\\
&\ge v_i^-(\smash{I^*_{j'}}) \qquad \quad &&\text{by Lemma~\ref{lem:ef1:approx-correct}},
\end{alignat*}
for any $j' \in [n]$. Thus, $(\pi_j)_{j \in [n]}$ certifies that $\calI^*$ is a secretive EF1$_{\emph outer}$ division.~\qed
\end{proof}

Similar to the previous proof, one can prove the existence of an extra EF1$_{\emph outer}$ connected division.

\begin{proof}[of Theorem~\ref{thm:EF1:extra}]
We use the same simplex $S_m$, triangulation $\ttT_{\half}$, and coloring function $\lambda_i$ $(i \in [n])$ as defined in Section~\ref{sec:EF1}. For agent $n+1$, we define the virtual valuation ${\hat v}_{n+1}(\bfx,j)$ and its coloring function $\lambda_{n+1}$ similarly as in Section~\ref{sec:EF1}. 
By applying Theorem \ref{thm:multisperner} to $\ttT_{\half}$ with $\lambda_i$ $(i \in [n+1])$, we obtain an elementary simplex $S^*_2=\langle \bfx^*_1,\bfx^*_2,\ldots,\bfx^*_n \rangle \in \ttT_{\half}$ satisfying Condition $($2$)$ of Theorem \ref{thm:multisperner}: that is, regardless of which $\lambda_{i^*}$ disappears, the graph $G(S^*_2)$ has a matching covering the vertices in $[n]$. Thus, for each $i^* \in [n+1]$, there exists a bijection $\pi_{i^*} \colon  [n+1]\setminus \{i^*\}  \rightarrow [n]$ such that for every agent $i \in [n+1]\setminus \{i^*\}$,
\begin{align}\label{eq:extra}
&\pi_{i^*}(i) \in \lambda_{i}(\bfx^*_{k})~~\mbox{for some}~~\bfx^*_k.
\end{align}

By applying Algorithm \ref{alg:rounding} to $S^*_2$, we obtain a division $\calI^*=(I^*_1,I^*_2,...,I^*_n)$. Take any extra agent $i^* \in [n+1]$ and any remaining agent $i \in [n+1] \setminus \{i^*\}$. For a bijection $\pi_{i^*} \colon  [n+1]\setminus \{i^*\}  \rightarrow [n]$, there exists a main vertex $\bfx^*_k$ of $S^*_2$ where $\pi_{i^*}(i) \in \lambda_{i}(\bfx^*_{k})$. Hence,
\begin{alignat*}{2}
v_i(I^*_{\pi_{i^*}(j)})
&\ge \hat v_i (\bfx^*_k, \pi_{i^*}(i)) &&\text{by Lemma~\ref{lem:ef1:approx-correct}},
\\
&\ge \hat v_i (\bfx^*_k, j) &&\text{since $\pi_{i^*}(i) \in \smash{\lambda_{i}}(\bfx^*_{k})$},\\
&\ge v_i^-(\smash{I^*_{j}}) \qquad \quad  &&\text{by Lemma~\ref{lem:ef1:approx-correct}},
\end{alignat*}
for any $j \in [n]$. Thus, $(\pi_{i^*})_{i^* \in [n+1]}$ certifies that $\calI^*$ is a extra EF1$_{\emph outer}$ division.~\qed 
\end{proof}


\section{Conclusion and discussion}
We proved that under connectivity constraints, an EF1 division exists for any number of agents with monotone valuations, thereby resolving the open problem raised by \citet{BiloCaFl19}. We further extended this existential result to the secretive and extra variants. 

In contrast with the standard existence result in cake-cutting \citep{Su1999,MeunierSu}, our proof requires monotonicity in the agents' valuations. An interesting open question is whether the hungry preference assumption---where agents always prefer any non-empty bundle to an empty bundle---is sufficient to prove the existence of an EF1 connected division for any number of agents. 

Recent research on fair division extensively investigates the setting where agents may have both positive and negative values for the items  \citep{AzizCaIg19,MeSh19,Halevi2018,brczi2020envyfree,jojic2019splitting}. In particular, \citet{AzizCaIg19} proposed an extension of EF1 to this more general setting, requiring agents' envy to disappear after the removal of one chore from an envious bundle or that of one good from an envied bundle. It will be interesting to investigate whether such a fairness notion can be achieved under connectivity constraints of a path. 
A possible direction would be to develop a similar discretization technique of the topological proof provided by \citet{jojic2019splitting}, who showed that an envy-free division of a partially burned cake exists when the number of agents is a prime power; see also \citet{Halevi2018} and \citet{MeSh19}.

Finally, this study highlights that the complexity of finding an EF1 connected division is an open problem. In particular, it would be interesting to settle the complexity question for a simple class of valuations, e.g. binary additive valuations.


\section*{Acknowledgments}
The author thanks Fr{\'{e}}d{\'{e}}ric Meunier, Dominik Peters, Warut Suksompong, and William S. Zwicker for valuable insights and feedback.

\bibliographystyle{plainnat}

\end{document}